\documentclass[a4paper,UKenglish,cleveref, autoref, thm-restate]{lipics-v2021}

\usepackage{cite}
\usepackage{bm, amsmath}

\usepackage{tikz}
\usetikzlibrary{arrows,arrows.meta,bending,shapes,calc,intersections,positioning}
\tikzstyle{bank}=[circle, draw, inner sep = 1]

\newcommand{\N}{\mathbb{N}}
\newcommand{\R}{\mathbb{R}}

\newcommand{\veca}{\mathbf{a}}
\newcommand{\vecf}{\mathbf{f}}
\newcommand{\vecp}{\mathbf{p}}

\newcommand{\oc}{\text{oc}}

\newcommand{\vecEll}{\bm{\ell}}

\newcommand{\calF}{\mathcal{F}}

\newcommand{\knapsack}{\textsc{Knapsack}}
\newcommand{\subsetSum}{\textsc{Subset Sum}}
\newcommand{\setPacking}{\textsc{Set Packing}}
\newcommand{\independentSet}{\textsc{Independent Set}}
\newcommand{\MaxER}{\textsc{OutgoingER-VR}}
\newcommand{\MaxProp}{\textsc{OutgoingPROP-VR}}

\newcommand{\classP}{\textsf{P}}
\newcommand{\classNP}{\textsf{NP}}

\newcommand{\deb}{\text{de}}
\newcommand{\cre}{\text{cr}}

\newcommand{\ret}{\text{ret}}

\allowdisplaybreaks

\title{Algorithms for Claims Trading}

\author{Martin Hoefer}{Goethe University Frankfurt, Germany}{mhoefer@em.uni-frankfurt.de}{https://orcid.org/0000-0003-0131-5605}{}{}

\author{Carmine Ventre}{King's College London, United Kingdom}{carmine.ventre@kcl.ac.uk}{https://orcid.org/0000-0003-1464-1215}{}{}

\author{Lisa Wilhelmi}{Goethe University Frankfurt, Germany}{wilhelmi@em.uni-frankfurt.de}{https://orcid.org/0000-0003-0845-1941}{}{}

\authorrunning{M. Hoefer, C. Ventre, L. Wilhelmi} 

\Copyright{Martin Hoefer, Carmine Ventre, Lisa Wilhelmi} 



\category{} 

\relatedversion{} 

\acknowledgements{This paper was initiated at the Dagstuhl seminar 22452 ``Computational Social Dynamics''. The authors thank the organizers and participants of the seminar. Martin Hoefer and Lisa Wilhelmi gratefully acknowledge support by DFG Research Unit ADYN under grant DFG 411362735.}

\nolinenumbers 

\newcommand{\myparagraph}[1]{\medskip \noindent \textbf{\sffamily #1.}$\;$ }

\ArticleNo{0}

\begin{document}




\maketitle

\begin{abstract}
    The recent banking crisis has again emphasized the importance of understanding and mitigating systemic risk in financial networks. In this paper, we study a market-driven approach to rescue a bank in distress based on the idea of \emph{claims trading}, a notion defined in Chapter 11 of the U.S.\ Bankruptcy Code. We formalize the idea in the context of the seminal model of financial networks by Eisenberg and Noe~\cite{EisenbergSystemicRisk}. For two given banks $v$ and $w$, we consider the operation that $w$ takes over some claims of $v$ and in return gives liquidity to $v$ (or creditors of $v$) to ultimately rescue $v$ (or mitigate contagion effects). We study the structural properties and computational complexity of decision and optimization problems for several variants of claims trading.
    
    When trading incoming edges of $v$ (i.e., claims for which $v$ is the creditor), we show that there is no trade in which \emph{both banks $v$ and $w$ strictly} improve their assets. We therefore consider \emph{creditor-positive} trades, in which $v$ profits strictly and $w$ remains indifferent. For a given set $C$ of incoming edges of $v$, we provide an efficient algorithm to compute payments by $w$ that result in a creditor-positive trade and maximal assets of $v$. When the set $C$ must also be chosen, the problem becomes weakly \classNP-hard. Our main result here is a bicriteria FPTAS to compute an approximate trade, which allows for slightly increased payments by $w$. The approximate trade results in nearly the optimal amount of assets of $v$ in any exact trade. Our results extend to the case in which banks use general monotone payment functions to settle their debt and the emerging clearing state can be computed efficiently.
    
    In contrast, for trading outgoing edges of $v$ (i.e., claims for which $v$ is the debtor), the goal is to maximize the increase in assets for the creditors of $v$. Notably, for these results the characteristics of the payment functions of the banks are essential. For payments ranking creditors one by one, we show \classNP-hardness of approximation within a factor polynomial in the network size, in both problem variants when the set of claims $C$ is part of the input or not. Instead, for payments proportional to the value of each debt, our results indicate more favorable conditions. 
\end{abstract}


\section{Introduction}
\label{sec:intro}
The global banking crisis of March 2023 caused turmoil in a market fearful of the repeat of the Great Financial Crisis of 2007. These recent events serve as a stark reminder of the paramount importance of the study of systemic risk in financial networks. In this growing body of work, the focus is mainly on the complexity of computing \emph{clearing states}, known to measure the exposure of the different banks in the network to insolvencies within, see, e.g.,  \cite{EisenbergSystemicRisk,SchuldenzuckerS17,IoannidisKV22}, and strategic aspects of the banks' behavior, cf.  \cite{BertschingerStrategicPayments,PappW20,HoeferW22,KanellopoulosKZ21}. However, to calm the market and prevent contagion, regulators and central banks are more interested in finding ways to rescue banks in distress, reassure investors that the system is stable and avoid further bank runs. In fact, Silicon Valley Bank, Signature Bank and Credit Suisse -- the three banks at the heart of the crisis last March -- were all acquired by other banks in the network, and, by modifying the network, this has seemingly mitigated systemic risk. 

A line of research in financial networks on interventions in the network is recently discussed in \cite{PappDebtSwapping,froese2023complexity}, the main idea being that banks can swap debt contracts. In particular, the authors of \cite{froese2023complexity} study the extent to which a sequence of debt swaps can reduce the risk in the network, in the sense that bank assets Pareto-improve. Notably, swaps can occur anywhere in the network, even if the focus is strict improvement of the assets of a given bank.

In this work, we build on this idea and initiate research on the computation of a network-based ``rescue package'' deal for a given bank with the objective of making it solvent. This is exactly the problem that regulators faced in March 2023 for the aforementioned banks. 
However, \emph{acquisitions} do not seem to be the right operations in these instances since they have two main drawbacks from a societal perspective (as also witnessed by the reactions to recent deals). Firstly, the acquiring bank rarely has enough time or freedom to evaluate the purchase and make a sensible business decision. Secondly, and consequently, it often requires a security for bailout from the central bank, in the form of significant protection against potential losses from risks associated with the transaction.
For example, in the acquisition of Credit Suisse, UBS had little choice in the matter, as reported by Bloomberg news~\cite{Bloomberg23}, and received a guarantee worth CHF 9 billion, as confirmed by the Swiss Federal Council~\cite{Swiss23}.

We instead study a market-driven approach to rescue banks in distress based on the idea of \emph{claims trading}. Claims trading is defined in Chapter 11 of the U.S.\ Bankruptcy Code. We formalize the idea and analyze the consequences of such trades in the context of financial networks. When a company is in financial distress, its creditors can assert their rights to repayment by submitting a claim. At this point, a creditor can either wait for the positions to unwind and get (a part of) the claim once the bankruptcy is settled, or she can sell her credit claim to a willing buyer for some immediate liquidity. The former approach is equivalent to the mainstream work on systemic risk since the insolvency of a bank can directly cascade through the network via lower payments to its creditors. We want to explore ways to find interested buyers that purchase the claims of an insolvent bank $v$ and give liquidity to the network that ultimately rescues $v$. Ideally, the buyers should avoid any loss so that the cash invested in buying the claim will return via increased payments within the network; this way incentives of buyers are aligned, and systemic risk is reduced at no extra cost to the network. 

We design efficient algorithms to compute claims trades or settle the inherent complexity status of the problems. The importance of algorithms computing claim trades that resolve complicated systemic issues in finance cannot be underestimated. In practice, deals are concocted when markets are closed, and algorithms that efficiently compute solutions in these pressurised situations become essential. 

\myparagraph{Related Work}
%
Much of the work on systemic risk in financial networks, including ours, builds upon \cite{EisenbergSystemicRisk}. In this seminal work, the authors propose a model and prove existence and properties of clearing states. Moreover, they also provide a polynomial-time algorithm for their computation. The model in~\cite{EisenbergSystemicRisk} has been extended along many dimensions by follow-up work; for example, the authors of \cite{rogers2013failure} add default costs whereas financial derivatives are considered in \cite{schuldenzucker2016clearing}. Computation of clearing states for the latter model is studied in \cite{SchuldenzuckerS17,IoannidisKV22,IoannidisKV22b} for different notions of approximation and payment schemes adopted. The solution space of clearing states for financial networks with derivatives is studied in \cite{PappW21}.

The study of strategic behavior in financial networks was initiated in \cite{BertschingerStrategicPayments}, where banks are assumed to strategize in the way they allocate money to their creditors. A similar approach is used in \cite{KanellopoulosKZ21,HoeferW22}. A different model, featuring derivatives, and banks strategically donating money or cancelling debts is studied in \cite{PappW20}. The idea of cancelling debts is further explored in \cite{KanellopoulosKZ22}. The authors of~\cite{KanellopoulosKZ22,papachristou2022allocating} consider computational complexity of computing optimal or approximate bailout policies from the central bank external to the network. In contrast, in our work all transfers of assets are intrinsic to the network, and the bank providing the assets must not be harmed. In~\cite{kanellopoulos2023debt}, the authors study computational complexity of strategic changes to the underlying network via debt transfers. 

Debt swapping is introduced in \cite{PappDebtSwapping} -- the authors focus more on the existence and properties of swaps with and without shocks to the system. As discussed above, the authors of \cite{froese2023complexity} share goals that are somewhat similar to ours but use a different operation to update the network. A related line of work considers portfolio compression, an accounting operation by means of which all the cycles in the network are deleted. The effects on systemic risk of portfolio compression are studied in \cite{SchuldenzuckerS20,Veraart22}. However, it is important to note that portfolio compression can lead to a worse outcome for banks that are not contained in the cycle \cite{SchuldenzuckerS20} and consequently it is not clear why banks should accept to modify their balance sheets in this way, as argued empirically in \cite{MayoW21}.

To the best of our knowledge, ours is the first work to study claims trading in the analysis of financial networks. Claims trading in bankruptcy has been studied by law scholars, who for example argue that its effects in that context are variegated and nuanced in general \cite{Levitin09} but do not concern the governance of the bankruptcy process \cite{Ellias18}.

\myparagraph{Contribution}
We focus on the elementary setting with one given bank $v$ to save (e.g., Credit Suisse) and one bank $w$ that may rescue it (e.g., UBS). We consider the following problem: Are there claims of $v$ that can be sold to $w$ so that $v$ becomes solvent (i.e., after the claims trade, $v$ can fully pay all its liabilities)? This problem gives rise to a suite of algorithmic questions, depending on the remit of the algorithmic decision, such as: How many claims are we allowed to trade? Which claims of $v$ should we trade? What are the payments that must be transferred from $w$ to $v$ to make the trade worthwhile for $w$? Our treatment is steered by the following structural insight: We prove that it is impossible for \emph{both $v$ and $w$ to strictly} profit from the claims trade. Accordingly, we restrict our attention to creditor-positive trades that strictly improve $v$ without harming $w$. 

For our first set of algorithmic results in Section \ref{sec:single}, we fix one claim with creditor $v$ to be traded with $w$. Does this represent a feasible (i.e., creditor-positive) trade? This can be decided by simply computing {clearing states} that determine the payment towards each debt in the network. The problem becomes interesting if we also determine the \emph{haircut rate} $\alpha \in [0,1]$ of the trade -- in order to provide liquidity, $w$ may be willing to pay an $\alpha$-fraction of the claim's liability to $v$. Depending on the payment functions used by banks to distribute money to their debts, we design different polynomial-time algorithms that determine feasibility of the trade and also $\alpha^*$ (if any), the value of $\alpha$ maximizing the assets of $v$ (or a close approximation of $\alpha^*$ if the payment functions are too granular vis-a-vis the input size). Let us highlight that these results also apply to the case in which $v$ is the \emph{debtor} of the claim to trade; in fact, we prove that every creditor-positive trade Pareto-improves the clearing state -- each bank in the network is (weakly) better off after the trade. By maximizing the assets of the creditor of the traded claim we, thus, also maximize assets of the debtor. 

We consider trading multiple claims with creditor $v$ in Section~\ref{sec:multiInc}. For a fixed set of claims our results from Section~\ref{sec:single} extend rather directly. The picture becomes less benign when we also have to choose the subset of claims to be traded. Indeed, in Section~\ref{sec:multiIncChoose} we show that it is weakly \classNP-hard to decide if there is a subset of claims along with suitable haircut rates to obtain a creditor-positive claims trade that makes $v$ solvent. In our most technical contribution, we show that there exists a bicriteria FPTAS for deciding this problem. If an exact trade exists that yields total assets of $A^*$ for $v$, we find an approximate trade with assets at least $A^*-\delta$ for exponentially small $\delta$, which allows haircut rates of at most $1+\varepsilon$. The FPTAS applies to all financial networks with general monotone payment functions for which a clearing state can be computed efficiently. On a technical level, we fix a desired value $A$ for the total assets of $v$. Using a subroutine we determine if there is an approximate trade that yields this asset level for $v$. En route, we discover an intricate monotonicity property -- if there is an exact trade that yields assets $A^*$ for $v$, then for every $A \le A^*$ there is an approximate trade with assets $A$ for $v$. Notably, monotonicity can break above $A^*$. Still, we can apply binary search to find an approximate trade with assets at least $A^*-\delta$. 

Finally, for trading multiple claims with debtor $u$ in Section~\ref{sec:multiOut} -- rather than trying to save $u$ -- our goal is improve conditions for the creditors of $u$ to minimize the contagion effects by $u$'s bankruptcy. Interestingly, the results here depend significantly on the choice of the payment functions. For payments based on a ranking of the creditors, we show that the problem becomes \classNP-hard to approximate within a factor polynomial in the network size. In contrast, for payments proportional to the value of each debt, we can solve the problem for a given set of claims, but it becomes strongly \classNP-hard when having to choose the set of claims. 

All missing proofs are deferred to the Appendix~\ref{app:proofs}. For our results in the paper we assume that haircut rates can be chosen; Appendix~\ref{app:fixedAlpha} contains similar results for fixed haircut rates.


\section{Model and Preliminaries}
\label{sec:pre}

A financial network $\calF = (G, \vecEll, \veca^x, \vecf)$ is expressed as a directed multigraph%
\footnote{Claims trades in simple graphs can result in graphs with multi-edges. This can sometimes be avoided by analyzing the trades in equivalent simple graphs with suitable auxiliary banks. Since all our arguments can also be applied in the context of multigraphs, we discuss the more general model.} %
$G=(V,E,\deb,\cre)$ without self loops. We denote $n = |V|$. Every node $v \in V$ in the graph represents a financial institution or \emph{bank}. Every edge $e \in E$ represents a \emph{debt contract} or \emph{claim} involving two banks. For each edge $e \in E$, $\deb(e)$ specifies the debtor (i.e., the source) and $\cre(e)$ the creditor (i.e., target). Edge $e \in E$ has a weight $\ell_e \in \N_{>0}$. In other words, in the context of debt contract $e$, bank $\deb(e)$ owes $\cre(e)$ an amount of $\ell_e$. We denote the set of outgoing and incoming edges of a bank $v$ by $E^+(v) = \{e \in E \mid v = \deb(e)\}$ and $E^-(v) = \{e \in E \mid v = \cre(e)\}$. Since we allow multi-edges, several debt contracts with possibly different liabilities could exist between the same pair of banks. The \emph{total liabilities} $L_v$ of $v$ are the sum of weights of all outgoing edges of $v$, i.e., $\sum_{e \in E^+(v)} \ell_e = L_v$. Furthermore, every bank $v$ holds \emph{external assets} $a_v^x \in \N$. They can be interpreted as an amount of money the bank receives from outside the network.

Let $b_v \in [a_v^x, a_v^x + \sum_{e \in E^-(v)} \ell_e]$ be the total funds of bank $v$. Bank $v$ distributes her total funds according to a given \emph{payment function} $\vecf_v = (f_e)_{e \in E^+(v)}$, where $f_e: \R \rightarrow [0,\ell_e]$. For every outgoing edge, the function $f_e(b_v)$ defines the amount of money $v$ pays towards $e$. We follow previous literature and assume the following conditions for every payment function:
\begin{enumerate}
    \item[(1)] Every function $f_e(b_v)$ is non-decreasing and bounded by $0 \le f_e(b_v) \le \ell_e$.
    \item[(2)] Every bank pays all funds until all liabilities are settled: $\sum_{e \in E^+(v)}f_v(b_v) = \min\{b_v, L_v\}$.
    \item[(3)] The sum of payments of a bank is limited by the total funds: $\sum_{e \in E^+(v)} f_v(b_v) \leq b_v$. 
\end{enumerate}
Here (2) implies (3), and we mention (3) explicitly for clarity. For a \emph{monotone} function $\vecf_v$, $v$ weakly increases the payment on every outgoing edge when receiving additional funds. 

\myparagraph{Clearing States} Let $\vecp=(p_e)_{e \in E}$ be the arising payments in the network when every bank $v$ distributes the funds according to her payment functions $\vecf_v$. The \emph{incoming payments} of $v$ are given by $\sum_{e \in E^-(v)} p_e$. The \emph{total assets} $a_v$ are defined as the external assets plus the incoming payments, i.e., $a_v = a_v^x + \sum_{e \in E^-(v)} p_e$. Observe that the above conditions (1), (2) and (3) are fixed-point constraints. A vector of total assets $\veca = (a_v)_{v\in V}$ is called \emph{feasible} if it satisfies all fixed-point constraints. More formally, for every feasible $\veca$ it holds that $a_v = a_v^x + \sum_{e \in E^-(v)} f_{e} (a_{\deb(e)})$. The payments $\vecp$ corresponding to a feasible vector $\veca$ are called a \emph{clearing state}. For fixed payment functions, multiple clearing states may exist. We assume throughout that every payment function $\vecf_v$ is \emph{monotone}, i.e., $f_e(x) \ge f_e(y)$ for all $x \ge y \ge 0$ and every $e \in E^+(v)$. This implies that all clearing states form a complete lattice~\cite{BertschingerStrategicPayments, CsokaDecentralizedClearing}. Thus, the point-wise minimal and maximal clearing states are unique. We follow previous literature and assume that the \emph{maximal} clearing state arises in the network.

\myparagraph{Payment Functions} In the seminal work of Eisenberg and Noe \cite{EisenbergSystemicRisk} and the majority of subsequent works, all banks are assumed to allocate their assets using \emph{proportional} payment functions. The \emph{recovery rate} $r_v = \min \{a_v/L_v, 1\}$ is the fraction of total liabilities $v$ can pay off, and the payments on edge $e \in E^+(v)$ are defined proportionally by $f_e(a_v) = r_v \cdot \ell_e$. Hence, if $r_v = 1$, then $v$ will fully settle all liabilities. Otherwise, $v$ is in default, $r_v < 1$, and the liabilities are settled partially in proportion to their weight. These payments are often used when all debt contracts fall due at the same date. If, on the other hand, different debt contracts are assigned different priorities or maturity dates, payments are more suitably expressed by \emph{edge-ranking} payment functions. Then, the debt contracts in $E^+(v)$ are ordered by a permutation $\boldsymbol{\pi}_v$. First, $v$ makes payments towards the highest ranked edge $\pi_v(1)$ until the edge is saturated or $v$ has no remaining assets. Once $\pi_v(1)$ is fully paid off, $v$ pays off the second highest ranked edge $\pi_v(2)$ until the edge is saturated or $v$ has no remaining assets. The process continues and ends when either all liabilities are settled or $v$ exhausted all assets. 

Both proportional and edge-ranking payments are monotone. For proportional payments, the clearing state can be computed in polynomial time~\cite{EisenbergSystemicRisk}; for edge-ranking payments in strongly polynomial time~\cite{BertschingerStrategicPayments}.

In this paper, we obtain some results explicitly for networks with proportional and edge-ranking payment functions. Most of our results, however, generalize to arbitrary monotone payment functions when there is an efficient \emph{clearing oracle}, i.e., there exists an algorithm that receives a network $\calF$ as input and outputs the clearing state $\vecp$ in polynomial time.

\myparagraph{Claims Trades} When a bank $u$ is in default and unable to settle all debt, this introduces risk into the network. In particular, the creditors of $u$ do not receive their full liabilities. This could lead to further defaults of the creditor banks. In order to reduce the risk of spreading default, the creditors of $u$ can sell claims they make towards $u$. More in detail, consider banks $u,v$ and $w$ with edge $e$ with $\deb(e) = u$ and $\cre(e) = v$, $\ell_{e} \ge 0$. Suppose $u$ is in default. If $v$ and $w$ perform a claims trade, $w$ becomes the new creditor of bank $u$ with the same liability. Consequently, any payment from $u$ towards the claim will be received by $w$. In return for the traded claim, $v$ receives a \emph{return} $\rho$ from $w$, i.e., an immediate payment of $\rho = \alpha \cdot \ell_e$, for some $\alpha \in [0,1]$. We call $\alpha$ the \emph{haircut rate}.
To separate the return from the payments in the clearing state, we model the return by a transfer of external assets from $w$ to $v$. Note that $w$ can invest at most her external assets as a return, so every trade must satisfy $\alpha \ell_e \le a_w^x$. After a trade the external assets of $v$ and $w$ might no longer be integer values.

We proceed to define three variants of claims trades. We are given a financial network with distinct banks $u,v,w \in V$, an edge $e \in E$ with $(\deb(e),\cre(e)) = (u,v)$ and haircut rate $\alpha \in [0,1]$. For a \emph{(single) claims trade} of $e$ to $w$ we perform the following adjustments to the network: (1) change the creditor of $e$ to $\cre'(e) = w$, (2) change external assets of $v$ to $a^x_v + \alpha \cdot \ell_{e}$ and (3) change external assets of $w$ to $a^x_w - \alpha \cdot \ell_{e} \geq 0$. We denote the resulting \emph{post-trade network} by $\calF' = (G',\vecEll,\veca'^x,\vecf)$, and the resulting clearing state in $\calF'$ by $\vecp'$. For a given trade of $e$ to $w$, we call $v$ the \emph{creditor} and $w$ the \emph{buyer}. Observe that the total assets of $v$ after the trade are given by $a'_v = a^x_v + \alpha \cdot \ell_{e} + \sum_{e' \in {E'}^-(v)} p'_{e'}$. Similarly, the total assets of $w$ after the trade are $a'_w = a^x_w - \alpha \cdot \ell_{e} + \sum_{e' \in {E'}^-(w)} p'_{e'}$.

The claims trade operation can be directly extended to a trade of multiple edges. As outlined in the introduction, we are interested in the effects when a single bank (in default) trades claims with another bank $w$ (such as a central bank). We study the differences when trading incoming or outgoing edges. Observe that both generalize single claims trades.

For a \emph{multi-trade of incoming edges}, there are distinct banks $v,w$ in a network $\calF$, a set $C$ of $k$ distinct incoming edges $e_1,\ldots,e_k \in E^-(v)$, and haircut rates $\alpha_1,\ldots,\alpha_k$, such that $\deb(e_i) \neq w$, for all $i$. After the trade, a new network $\calF'$ emerges: We change $\cre'(e_i) = w$, for all $i=1,\ldots,k$, adjust external assets for $v$ to $a^x_{v} + \sum_{i=1}^k \alpha_i \ell_{e_i}$, and for $w$ to $a^x_w - \sum_{i=1}^k \alpha_i \ell_{e_i} \ge 0$.

For a \emph{multi-trade of outgoing edges}, there are distinct banks $u,w$ in a network $\calF$, a set $C$ of $k$ distinct outgoing edges $e_1,\ldots,e_k \in E^+(v)$ with $\cre(e_i)=v_i$, and haircut rates $\alpha_1,\ldots,\alpha_k$, such that $\cre(e_i) \neq w$, for all $i$. After the trade, a new network $\calF'$ emerges: We change $\cre'(e_i) = w$, for all $i=1\ldots,k$, adjust external assets for each $v_i$ to $a^x_{v_i} + \alpha_i \ell_{e_i}$, and for $w$ to $a^x_w - \sum_{i=1}^k \alpha_i \ell_{e_i} \ge 0$.

We proceed with a small example of trading a single claim.
\begin{example}\label{ex:intro}
    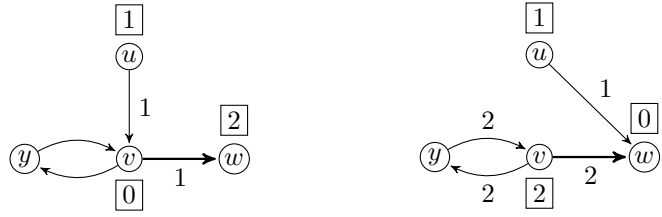
\begin{figure}
    \centering
    \begin{tikzpicture}[>=stealth', shorten >=1pt, auto, node distance=1cm, transform shape, align=center]
        \node[bank] (v) at (0,0) {$v$};
        \node[bank] (u) [above=of v] {$u$};
        \node[bank] (y) [left=of v] {$y$};
        \node[bank] (w) [right=of v] {$w$};
        \node[rectangle, draw=black, inner sep=2.5pt] (xv) [below=0.1cm of v] {$0$};
        \node[rectangle, draw=black, inner sep=2.5pt] (xu) [above=0.1cm of u] {$1$};
        \node[rectangle, draw=black, inner sep=2.5pt] (xw) [above=0.1cm of w] {$2$};
        \draw[->] (u) -- node[right] {$1$} (v);
        \draw[->] (v) to[bend left] (y);
        \draw[->] (y) to[bend left] (v);
        \draw[->, line width= 0.03cm] (v) -- node[below] {$1$} (w);
    \end{tikzpicture}
    \hspace{2cm}
    \begin{tikzpicture}[>=stealth', shorten >=1pt, auto, node distance=1cm, transform shape, align=center]
        \node[bank] (v) at (0,0) {$v$};
        \node[bank] (u) [above=of v] {$u$};
        \node[bank] (y) [left=of v] {$y$};
        \node[bank] (w) [right=of v] {$w$};
        \node[rectangle, draw=black, inner sep=2.5pt] (xv) [below=0.1cm of v] {$2$};
        \node[rectangle, draw=black, inner sep=2.5pt] (xu) [above=0.1cm of u] {$1$};
        \node[rectangle, draw=black, inner sep=2.5pt] (xw) [above=0.1cm of w] {$0$};
        \draw[->] (u) -- node[right=0.2cm, pos=0.3] {$1$} (w);
        \draw[->] (v) to[bend left] node[below] {$2$} (y);
        \draw[->] (y) to[bend left] node[above] {$2$} (v);
        \draw[->, line width= 0.03cm] (v) -- node[below] {$2$} (w);
    \end{tikzpicture}
    \caption{The network from Example~\ref{ex:intro} before the trade is depicted left, and right after the trade. All liabilities equal 2. Edges are labeled with positive payments (if any) in the clearing state.}
    \label{fig:example}
    \end{figure}
    Consider the example network depicted in Figure~\ref{fig:example} (left) on a simple directed graph. The liability of every edge is 2. The only banks with non-zero external assets are $u$ and $w$, where $a^x_u = 1$ and $a^x_w=2$. $a^x_v = 0$ is also explicitly displayed for convenience. Banks $u,w$ and $y$ each have at most one outgoing edge. They pay all their assets (if any) to the unique outgoing edge until it is saturated. This implies payments of $1$ on edge $(u,v)$. $v$ is the only bank with a non-trivial payment function -- suppose it uses an edge-ranking function with priority $\pi_v(1)=(v,w)$ and $\pi_v(2)=(v,y)$. Then, $v$ pays the incoming assets of $1$ to $w$, and there are no payments on the cycle of $v$ and $y$. To see this, assume $p_{(y,v)}= x > 0$. By the edge-ranking function, from these additional assets $v$ allocates a portion of $\min(x,1)$ to $(v,w)$ and the rest to $(v,y)$. Hence, the total assets of $y$ are $\max(x-1,0)$ while the outgoing payments are $x$, which contradicts the feasibility constraint (3). Overall, in the clearing state, the total assets are $a_u=a_v=1, a_w=3$ and $a_y=0$. 

    Suppose we perform the trade of edge $e = (u,v)$ to $w$ with $\alpha=1$, see Fig.~\ref{fig:example} (right) for the resulting network. $w$ buys edge $e$ and pays the full liability $\ell_e$ to $v$. The external assets of $v$ increase to 2 and allow $v$ to settle all debt. The total assets become $1$ for $u$, $2$ for $y$, $3$ for $w$ and $4$ for $v$. The total assets of $u$ and $w$ are unaffected by the trade, the total assets of $v$ and $y$ strictly increase. Overall, the clearing state is point-wise non-decreasing.

    A similar observation can be made when $v$ uses proportional payments. Before the trade, $v$ pays 1 to $y$ and $w$. After the trade, both edges can be paid fully. \hfill $\blacksquare$
\end{example}

\myparagraph{Properties of Claims Trades} In Example~\ref{ex:intro}, $v$ strictly benefits from the trade while $w$ is indifferent. Interestingly, it is impossible for both creditor and buyer to strictly profit from a single trade. This property holds true for the more general class of multi-trades of incoming edges, and it applies in any network $\calF$ with monotone payment functions.

The proof builds on a connection with \emph{debt swaps} studied in~\cite{PappDebtSwapping, froese2023complexity}. A debt swap exchanges the creditors of two edges with the same liabilities. We show that a claims trade and the resulting payments can be represented by a debt swap in an auxiliary network.  
\begin{definition}[Debt Swap]
    Consider a financial network $\calF$ with four distinct nodes $u_1,u_2,v_1,v_2 \in V$ and edges $e_1,e_2 \in E$, where $u_1=\deb(e_1), v_1=\cre(e_1)$ and $u_2=\deb(e_2),v_2=\cre(e_2)$. Suppose the liabilities are $\ell_{e_1}=\ell_{e_2}$. A debt swap $\sigma$ of $e_1$ and $e_2$ creates a new network $\calF^{\sigma}$ with $G^{\sigma}=(V,E, \deb,\cre^{\sigma})$ where $\cre^{\sigma}(e_1)=v_2$, $\cre^{\sigma}(e_2)=v_1$ and $\cre^{\sigma}(e)=\cre(e)$ for all $e \in E\setminus\{e_1,e_2\}$.
\end{definition}

\begin{proposition}\label{prop:positive}
    For every financial network with monotone payment functions, there exists no multi-trade of incoming edges such that both creditor $v$ and buyer $w$ strictly improve their total assets.
\end{proposition}
\begin{proof}
    For a given network $\calF$, consider a multi-trade of incoming edges and construct a new network $\hat{\calF}$ by adding an auxiliary bank $\hat{v}$ to $\calF$ without external assets. $\hat{v}$ serves as an ``accumulator'' for the payments along the edges $e_i$. We change the targets of the edges in $C$ to $\cre(e_i) = \hat{v}$. We add an edge $\hat{e}$ with $\deb(\hat{e}) = \hat{v}$ and $\cre(\hat{e}) = v$ and liability $\ell_{\hat{e}} = \sum_{i=1}^k \ell_{e_i}$. Every payment that gets paid to $v$ via $e_i$ in $\calF$ now first goes to $\hat{v}$, and then gets forwarded from $\hat{v}$ to $v$, since $\hat{e}$ has sufficiently high liability. Consider the clearing state $\hat{\vecp}$ in the resulting network $\hat{\calF}$. Obviously, for the new edge $\hat{p}_{\hat{e}} = \sum_{i=1}^k \hat{p}_{e_i}$. As such, every (non-auxiliary) bank from $\calF$ receives the same external assets and eventually the same incoming and outgoing payments in $\hat{\calF}$. Consequently, both $\calF$ and $\hat{\calF}$ give rise to the same clearing state, i.e., $p_e =  \hat{p}_e$, for all $e \in E$, and the same assets for every (non-auxiliary) bank. 

    The new network $\hat{\calF}$ allows to conveniently route all payments along edges in $C$ to $w$ by trading the single accumulator edge $\hat{e}$ to $w$. Thus, the multi-trade of incoming edges $C$ in $\calF$ is equivalent to trading the single claim $\hat{e}$ to $w$ in $\hat{\calF}$, for a suitably chosen haircut rate $\hat{\alpha}$ such that $\hat{\alpha} \cdot \sum_{i=1}^k\ell_{e_i} = \sum_{i=1}^k \alpha_i \cdot \ell_{e_i}$.

    Now let us further adjust $\hat{\calF}$ to $\tilde{\calF}$ by introducing an auxiliary bank $\tilde{w}$. Intuitively, we ``outsource'' parts of external assets from $w$ to $\tilde{w}$. Formally, external assets of $w$ are reduced to $a^x_w - \sum_{i=1}^k \alpha_i\cdot \ell_{e_i} \geq 0$, external assets of $\tilde{w}$ are $\sum_{i=1}^k \alpha_i \cdot \ell_{e_i}$. We add an edge $\tilde{e}$ with $\deb(\tilde{e}) = \tilde{w}$ and $\cre(\tilde{e}) = w$ as well as liabilities $\ell_{\tilde{e}} = \ell_{\hat{e}} = \sum_{i=1}^k \ell_{e_i}$. The clearing state $\tilde{\vecp}$ in the resulting network $\tilde{\calF}$ is $\tilde{p}_{\tilde{e}} = \sum_{i=1}^k \alpha_i \ell_{e_i}$, since $\tilde{e}$ is the only outgoing edge of $\tilde{w}$ and $\ell_{\tilde{e}} \geq a^x_{\tilde{w}}$. Hence, $w$ and (consequently) every non-auxiliary bank from $\calF$ receives the same total assets in $\tilde{\vecp}$. Indeed, $\calF$, $\hat{\calF}$ and $\tilde{\calF}$ yield equivalent clearing states with $p_e = \hat{p}_e = \tilde{p}_e$, for all $e \in E$. 
     
    In $\tilde{\calF}$ we can implement the return payments from $w$ to $v$ by re-routing the ``outsource'' edge $e'$ to $v$ instead of $w$. Thus, the claim trade of $\hat{e}$ in $\hat{\calF}$ can be expressed by a swap of creditors of $\hat{e}$ and $\tilde{e}$ in $\tilde{\calF}$. Now since $\ell_{\tilde{e}} = \ell_{\hat{e}}$, this swap of creditors represents a debt swap. Thus, the multi-trade in $\calF$ is equivalent to single trade in $\hat{\calF}$ and a debt swap in $\tilde{\calF}$. No debt swap can strictly improve both creditor banks~\cite[Corollary 6]{froese2023complexity}. Thus, no multi-trade of incoming edges can strictly improve both creditor and buyer.
\end{proof}
The above proof implies a structural equivalence. Using the network $\hat{\calF}$, we reduced a multi-trade of incoming edges to a single claim trade.
\begin{corollary}\label{cor:multi-to-single}
    For every multi-trade of incoming edges in a network $\calF$, there is an adjusted network $\hat{\calF}$ such that the result of the multi-trade in $\calF$ is the result of a single trade in $\hat{\calF}$.
\end{corollary}
Our motivation is to analyze claims trades to improve the situation of a creditor in default by trading claims with a buyer. Since it is impossible to strictly improve the conditions of both banks, we focus on strictly improving the creditor and weakly improving the buyer. Note that the trade performed in Example~\ref{ex:intro} satisfies this property.
\begin{definition}[Creditor-positive trade]~\label{def:creditPos}
    A multi-edge trade of incoming edges of bank $v$ to bank $w$ is called creditor-positive if $a'_v > a_v$ and $a'_w \geq a_w$.
\end{definition}
For the proof of Proposition~\ref{prop:positive}, we express the multi-trade by a debt swap in an auxiliary network. For a creditor-positive trade, the associated debt swap satisfies the same property, i.e., it is a so-called \emph{semi-positive} debt swap. In any network $\calF$ with monotone payment functions, a semi-positive debt swap Pareto-improves the clearing state and, hence, the total assets of \emph{every} bank~\cite{froese2023complexity}. This directly implies the next corollary. 
\begin{corollary}
    \label{cor:Pareto}
   In every financial network with monotone payments, every creditor-positive trade Pareto-improves the clearing state. 
\end{corollary}
A creditor-positive trade reduces the impact of a defaulting debtor on the creditor. No bank in the entire network suffers. Hence, these trades contribute to the stabilization of the entire financial network. We focus on creditor-positive trades for the remainder of the paper.


\section{Trading a Single Claim}\label{sec:single}
In this section, we study a given single creditor-positive trade and optimize the effects on the assets in the network. For exposition, we mostly focus on financial networks with proportional or edge-ranking payments. There are two natural problem variants depending on whether the haircut rate $\alpha$ is fixed a priori or can be chosen as part of the trade. 

The variant with fixed $\alpha~\in [0,1]$ is simple (see Appendix~\ref{app:fixedAlpha}). Instead, suppose $\alpha$ (and, hence, return $\rho$) is chosen as part of the trade. 
The choice of $\alpha$ affects the external assets of $v$ and, thus, payments throughout the network. If a given trade is creditor-positive for some $\alpha \in [0,1]$, we say that $\alpha$ is creditor-positive. Can we efficiently decide the existence of a creditor-positive $\alpha$? What is the \emph{optimal} $\alpha$ to maximize the improvement $a'_v - a_v$ of $v$? Clearly, a trade with optimal $\alpha$ maximizes the total assets $a'_v$. Since $a'_w = a_w$ in every creditor-positive trade, maximizing $a'_v$ also maximizes the payments of $v$, the incoming payments of $v$'s creditors, and, inductively, the payments and assets of every edge and bank in the network. A creditor-positive $\alpha$ that maximizes $a'_v$ also simultaneously maximizes (1) the Pareto-improvement of payments for each edge in the network, and (2) the return $\rho$ by $w$. This holds for all networks with monotone payment functions.

To answer the above questions, we modify $\calF$ into a \emph{return network} $\calF^{\ret}$ defined as follows. We switch edge $e$ to $\cre(e) = w$ and add a \emph{return edge} $e_r$ with $\deb(e_r) = w$ and $\cre(e_r) = v$. The payment on this edge models the return from $w$ to $v$, so the liability is $\ell_{e_r} = \min\{\ell_e, a^x_w\}$. Since we consider creditor-positive trades, we modify the payment function of $w$ as follows. For all $e' \in E^+(w) \setminus \{e_r\}$, we set $f^{\ret}_{e'}(x) = f_{e'}(x)$ for all $x \le a_w$ and $f^{\ret}_{e'}(x) = f_{e'}(a_w)$ for all $x \ge a_w$. For $e_r$ we set $f^{\ret}_{e_r}(x) = 0$ for all $x \le a_w$ and $f^{\ret}_{e_r}(x) = x-a_w$ for all $x \ge a_w$. Similarly, we modify the liabilities to $\ell_{e'} = f_{e'}(a_w)$. 
Intuitively, in $\calF^{\ret}$ $w$ maintains its payments up to a total outgoing assets of $a_w$. It forwards any assets exceeding $a_w$ as return via $e_r$ to $v$.

\begin{lemma}
    \label{lem:returnNetwork}
    Consider the clearing state $\vecp^{\ret}$ in $\calF^{\ret}$. 
    \begin{itemize}
        \item[(a)] Suppose there is an optimal creditor-positive $\alpha$ with return $\rho = \alpha \ell_e$. Then $\calF^{\ret}$ has $a^x_w > p_e$. In $\vecp^{\ret}$ we obtain assets of $a^{\ret}_w \in (a_w, a_w + \ell_{e_r}]$ 
        and $a^{\ret}_v > a_v$, and $p^{\ret}_{e_r} = \rho$.
        \item[(b)] If $a^x_w > p_e$ and $\vecp^{\ret}$ yields assets of $a^{\ret}_w \in (a_w, a_w + \ell_{e_r}]$ 
        and $a^{\ret}_v > a_v$, then payment $p_{e_r}$ represents a return of an optimal creditor-positive trade. 
    \end{itemize}
\end{lemma}

\begin{proof}
    We first show (a). Suppose there is an optimal creditor-positive $\alpha$. It results in a return $\rho = \alpha \ell_e \le \min\{\ell_e, a^x_w\} = \ell_{e_r}$, assets of $a_w$ for $w$, and $a'_v > a_v$ for $v$. When we assign payments $\hat{p}_e = p'_e$ for all $e' \in E$ and set the payment on $e_r$ to $\hat{p}_{e_r} = \rho$, we obtain a vector of payments $\hat{\vecp}$ in $\calF^{\ret}$ that satisfies all fixed-point conditions.

    We first show that this implies $a^x_w > p_e$, the payment on $e$ in $\vecp$ before the trade. Consider the assets of $w$. We have $\hat{a}_v = a'_v > a_v$. Recall $\vecp' \ge \vecp$ by Corollary~\ref{cor:Pareto}, so
    \begin{align*}
        \hat{a}_v &= a_v^x + \rho + \sum_{e' \in E^-(v) \setminus \{e\}} \hat{p}_{e'} = a^x_v + \rho + \sum_{e' \in E^-(v) \setminus \{e\}} p'_{e'} \\
        &\ge a^x_v + \rho + \sum_{e' \in E^-(v) \setminus \{e\}} p_{e'} = a_v + \rho - p_e
    \end{align*}
    Hence $a^x_w \ge \rho > p_e$, as desired.

    For the other conditions, consider the clearing state $\vecp^{\ret}$ in $\calF^{\ret}$. Due to maximality of the clearing state, $\vecp^{\ret} \ge \hat{\vecp}$. Thus, $a_w^{\ret} > a_w$, $a^{\ret} > a_v$ and $p^{\ret}_{e_r} \ge \rho$. We show that, indeed, $\hat{\vecp} = \vecp^{\ret}$, and that the condition $a_w + \ell_{e_r} \ge \hat{a}_w$ is satisfied.
    \begin{description}
    \item[Case 1:] The clearing state satisfies $a_w + \ell_{e_r} \ge a^{\ret}_w$. Then we prove below that $\vecp^{\ret}$ is equivalent to a creditor-positive trade with payments that Pareto-dominate $\hat{\vecp}$ and, consequently, higher assets for $v$ with $\hat{a}_v \ge a^{\ret}_v$. As such, $\vecp^{\ret}$ represents a better creditor-positive trade, a contradiction to $\hat{\vecp}$ stemming from an optimal one.
    \item[Case 2:] The clearing state satisfies $a_w + \ell_{e_r} < a^{\ret}_w$. Then $a^{\ret}_v > a_v$, and $w$ is solvent in $\calF^{\ret}$. Indeed, $w$ could transfer even more assets to the edges of $E^+(w) \setminus \{e_r\}$. This implies that with return $\ell_{e_r}$, there is a clearing state in $\calF'$ that can \emph{strictly improve both} $v$ and $w$. This is a contradiction to Corollary~\ref{prop:positive}.
    \end{description}

    To prove (b), suppose $\vecp^{\ret}$ fulfills the conditions. Then, clearly, the payment $p^{\ret}_{e_r}$ represents a feasible return. The payments $p^{\ret}_{e'}$ on the other edges $e' \in E \setminus \{e_r\}$ fulfill the fixed-point conditions in $\calF'$. Now for contradiction assume that $p'_{e'} > p^{\ret}_{e'}$ for some $e'$. Then $e' \neq e_r$, since we assume $p^{\ret}_{e_r}$ is the return used to construct $\calF'$. Hence, any strict increase in $\vecp'$ could be manifested in $\vecp^{\ret}$ as well, which contradicts the maximality of $\vecp^{\ret}$ in $\calF^{\ret}$.
\end{proof}
\begin{corollary}
    \label{cor:monotone}
    Consider a given single claims trade of $e$ to $w$.
    \begin{enumerate}[(a)]
        \item A return of $a^x_w \ge \rho > p_e$ is necessary to make the trade creditor-positive. For $\rho = p_e$, we obtain $\vecp' = \vecp$.
        \item Consider all creditor-positive $\alpha$. A value $\alpha$ with return $\rho = \alpha \ell_e$ maximizes the assets of $v$ if and only if it maximizes the payment on every single edge in $\calF'$, the assets of every single bank, as well as the value of $\rho$ and $\alpha$. 
    \end{enumerate}
\end{corollary}

\begin{proposition}
    \label{prop:singleTradeRank}
    For a given financial network with edge-ranking payments and a single claims trade, there is an efficient algorithm to compute an optimal creditor-positive $\alpha^* \in [0,1]$ or decide that none exists.
\end{proposition}

\begin{proof}
    We construct network $\calF^{\ret}$ as described above. Observe that the adjusted payment function $\vecf^{\ret}_w$ is again an edge-ranking function -- it first fills edges according to $\vecf_w$ until assets $a_w$ are paid. Thus, at most one edge $e' \in E^+(w)$ is paid partially. For this edge, the liabilities are reduced to $f_{e'}(a_w)$. For all other edges, the liabilities either remain untouched or are decreased to 0. Then the additional assets are directed to $e_r$. Thus, $\vecf^{\ret}_w$ can be represented by the same ranking as $\vecf_w$ up to (and including) edge $e'$, and then using $e_r$ as the next (and last) edge in the order. Hence, we can compute $\calF^{\ret}$ in strongly polynomial time. By checking the conditions of Lemma~\ref{lem:returnNetwork}, we can verify in polynomial time whether or not a creditor-positive trade exists and obtain the optimal return as the payment on $e_r$.
\end{proof}

\begin{restatable}{proposition}{propLP}\label{prop:single-trade-prop}
    For a given financial network with proportional payments and a single claims trade, there is an efficient algorithm to compute an optimal creditor-positive $\alpha^* \in [0,1]$ or decide that none exists.
\end{restatable}

Finally, our main result in this section shows that for general monotone payments with efficient clearing oracle, we can obtain an approximately optimal solution via binary search.
\begin{theorem}\label{thm:singleTrade}
    Consider a given financial network with monotone payment functions and efficient clearing oracle. For a given single claims trade, there exists an additive FPTAS for approximating the optimal improvement of $v$ from any creditor-positive $\alpha$.
\end{theorem}

\newcommand{\spl}{\text{sp}}

Our algorithm uses binary search. Towards this end, we first show, for a given target value $A \ge a_v$, how to verify the existence of a trade that achieves at least a value $A$ for the total assets of $v$. For intuition, we use a \emph{split network} $\calF^{\spl}$ obtained from $\calF'$ after the trade as follows: We replace $v$ and $w$ by source and sink banks $v_{in}, v_{out}$, and $w_{in}, w_{out}$. $v_{in}$ has the incoming edges of $v$, $w_{in}$ the ones of $w$ (including $e$). The outgoing edges of $v$ ($w$) are attached to $v_{out}$ ($w_{out}$). We set the external assets of $v_{out}$ and $w_{out}$ to $A$ and $a_w$, and these banks use the payment functions of $v$ and $w$, respectively. As such, the clearing state $\vecp^{\spl}$ in $\calF^{\spl}$ can be computed using the clearing oracle. 

Consider the incoming payments of $\vecp^{\spl}$ at $v_{in}$ and $w_{in}$. These payments shall exactly recover the expenses at $v_{out}$ and $w_{out}$ -- modulo external assets and the return payment from $w$ to $v$. We define the \emph{budget difference} by
\[
    d^{\spl}_w = \left(a^x_w + \sum_{e' \in E^-(w_{in})} p^{\spl}_{e'}\right) - a_w \qquad \text{and} \qquad d^{\spl}_v = A - \left(a^x_v + \sum_{e' \in E^-(v_{in})} p^{\spl}_{e'}\right)\;.
\]
$d^{\spl}_w$ is the surplus money earned by $w_{in}$ that shall be invested into the return, $d^{\spl}_v$ is the excess money spent by $v_{out}$ that must be recovered through the return.
\begin{lemma}
    \label{lem:monoton}
    For a given single claims trade and a given target value $A > a_v$, there is a creditor-positive trade with a value at least $A$ for $v$ if and only if $d^{\spl}_v = d^{\spl}_w > 0$.
\end{lemma}
\begin{proof} 
    We first show that if $d^{\spl}_v = d^{\spl}_w > 0$, then there exists a creditor-positive trade with asset at least $A$ for $v$. Suppose we consider $\vecp^{\spl}$ in the network $\calF'$ using return $\rho = d^{\spl}_v = d^{\spl}_w$. This exactly equilibrates the budgets of $v$ and $w$ -- $v$ receives $d^{\spl}_v$, the money needed to obtain total assets of $A$. Also, $w$ spends exactly $d^{\spl}_w$, the money needed to obtain total assets of $a_w$. Hence, $\vecp^{\spl}$ satisfies all fixed-point conditions in $\calF'$. As such, $\vecp' \ge \vecp^{\spl}$ coordinate-wise due to maximality of the clearing state. This implies that using return $\rho$, the clearing state $\vecp'$ yields $a_v' \ge A > a_v$ and $a'_w \ge a_w$. A creditor-positive trade with return $\rho$ exists.

    Now for the other direction, consider an optimal creditor-positive trade, which yields the highest asset level $A^*$ and consider any $A \in (a_v, A^*]$. We show that in this case $d^{\spl}_v = d^{\spl}_w > 0$ holds in the clearing state $\vecp^{\spl}$ of $\calF^{\spl}$ with external assets $A$ for $v_{out}$.
    
    Consider the optimal trade, its return $\rho^* > 0$ and the emerging payments $\vecp^*$ in $\calF'$ after this trade. Now in the corresponding split network $\calF^{*,\spl}$ with external assets of $A^*$ for $v_{out}$, the payments $\vecp^*$ yield $d^*_v = d^*_w = \rho^*$, by definition of $\vecp^*$. The previous paragraph and maximality of $A^*$ then imply that $\vecp^*$ must also be the clearing state $\vecp^{*,\spl} = \vecp^*$ of $\calF^{*,\spl}$.

    Now suppose in $\calF^{*,\spl}$ we reduce the external assets of $v_{out}$ by $\varepsilon = A^* - A > 0$. Then $\calF^{\spl}$ evolves. Since we reduce the assets of a single source $v_{out}$ by $\varepsilon$, we obtain $\vecp^{*,\spl} \ge \vecp^{\spl}$. Moreover, by non-expansivity~\cite[Lemma 33]{froese2023complexity}, the total incoming assets of all sinks must reduce by \emph{at most} $\varepsilon$. For the sinks $v_{in}$ and $w_{in}$ we set
    \[
        \varepsilon_v = a^{*,\spl}_{v_{in}} - a^{\spl}_{v_{in}} = \sum_{e' \in E^-(v_{in})} p^*_{e'} - p^{\spl}_{e'} \hspace{1.5cm}
        \varepsilon_w = a^{*,\spl}_{w_{in}} - a^{\spl}_{w_{in}} = \sum_{e' \in E^-(w_{in})} p^*_{e'} - p^{\spl}_{e'}
    \]
    and, thus, $d_v^{\spl} = d^*_v - (\varepsilon - \varepsilon_v)$ and $d_w^{\spl} = d^*_w - \varepsilon_w$. Non-expansivity implies $\varepsilon_v + \varepsilon_w \le \varepsilon$. 

    First, we observe that $\varepsilon_v + \varepsilon_w < \varepsilon$ is impossible. Then $\varepsilon_w < \varepsilon - \varepsilon_v$, so $d_w^{\spl} > d_v^{\spl}$, i.e., $w$ has more excess money in $\vecp^{\spl}$ than required by $v$. Consider a return of $\rho = d_v^{\spl}$ and $\vecp^{\spl}$ as payment vector in the resulting network $\calF'$. Then all banks are feasible w.r.t.\ the fixed-point conditions, except for $w$ which has strictly more income than outgoing assets. Hence, the clearing state satisfies $\vecp' \ge \vecp^{\spl}$, $a'_v \ge A > a_v$, and $a'_w > a_w$, a contradiction to Proposition~\ref{prop:positive}.

    Second, suppose that $\varepsilon_v + \varepsilon_w = \varepsilon$, then $d_v^{\spl} = d_w^{\spl}$. Then the clearing state $\vecp^{\spl}$ exactly fulfills the fixed-point conditions for all banks in $\calF'$ and yields assets $A > a_v$ for $v$ and $a_w$ for $w$ with $\rho = d_v^{\spl}$. Note that $\rho > 0$, since otherwise we contradict the maximality of the initial clearing state $\vecp$. Therefore, the existence of a creditor-positive trade with assets $A^* > A$ for $v$ implies that $d_v^{\spl} = d_w^{\spl} > 0$ for $\vecp^{\spl}$ in $\calF^{\spl}$ emerging from $A$.
\end{proof}

We are now ready to prove Theorem~\ref{thm:singleTrade}.
\begin{proof}[Proof of Theorem~\ref{thm:singleTrade}]
    Our algorithm works by testing different target values $A$ for the total assets of $v$. For a given target value $A$, we then use Lemma~\ref{lem:monoton} to verify existence of a return $\rho$ achieving at least assets $A$ for $v$. The maximum achievable assets for $v$ are $M_v = \sum_{e' \in E^-(v) \setminus e} \ell_{e'} + a^x_v + \min\{a^x_w,\ell_e\}$. We determine the maximal achievable $A$ using binary search on the interval $(a_v,M_v]$. 

    More formally, we choose $\delta > 0$ and apply binary search over the set $\{a_v + \delta, a_v + 2\delta, \dots, M_v\}$. Verifying the condition in Lemma~\ref{lem:monoton} can be done in polynomial time via a call to the clearing oracle in $\calF^{\spl}$. If the algorithm discovers that the condition does not hold for all tested values, then no creditor-positive trade with asset level at least $a_v + \delta$ for $v$ exists. Otherwise, let $\hat{A}$ be largest discovered value for which the test is positive. Then, any value of at least $\hat{A} + \varepsilon$ cannot be achieved for any return $\rho$. Hence, the optimal achievable total assets of $v$ in any creditor-positive trade are bounded by $A^* \in [\hat{A} , \hat{A} + \delta]$, and the additive approximation follows $\hat{A} - a_v \ge (A^*-a_v) - \delta$.

    For the running time, we require at most $\lceil \log_2 (1 + (M_v - a_v)/\delta) \rceil$ oracle calls, which is polynomial in the input size and $1/\delta$.
\end{proof}

Since the number of possible (single) claims trades in a network is limited by $|E| \cdot |V|$, we can use the algorithm above to compute every creditor-positive trade with an (approximately) optimized haircut rate for a given network in polynomial time.


\newcommand{\vecAlpha}{\bm{\alpha}}

\section{Multi-Trades of Incoming Edges}
\label{sec:multiInc}

\subsection{Fixed Set of Claims}
\label{sec:multiIncFixed}
In this section, we are interested in multi-trades of incoming edges of a creditor bank $v$ to a buyer bank $w$. This arises naturally, for example, when a high fraction of $v$'s debtors are in default or $v$ is ``too big to fail''. Then bankruptcy of $v$ would cause significant damage to the entire network. 

We are given a financial network $\calF$ with two distinct banks $v$ and $w$, and a set $C$ of $k$ incoming edges of $v$. Suppose the haircut rates $\alpha_i$ can be chosen individually for each $e_i \in C$ as part of the trade (for fixed rates, see Appendix~\ref{app:fixedAlpha}). 
We call a vector $\vecAlpha = (\alpha_1,\dots, \alpha_k)$ of haircut rates \emph{creditor-positive} if the resulting multi-trade is creditor-positive. Our goal is to select creditor-positive $\alpha_i \in [0,1]$, for every $i \in [k]$, in order to maximize the improvement of $v$, i.e., $a'_v - a_v = \sum_{i=1}^k \alpha_i \cdot \ell_{e_i} + \sum_{e' \in E'^-(v)} p'_{e'} - \sum_{e' \in E^-(v)}p_{e'}$. Observe that we can restrict our attention to vectors with uniform $\alpha_i = \alpha'$ for all $i \in [k]$ and some $\alpha' \in [0,1]$ -- given any $\vecAlpha$, choose $\vecAlpha'$ with $\alpha'_i = \alpha'$ such that $\alpha' \cdot \sum_{i=1}^k \ell_{e_i} = \sum_{i=1}^k \alpha'_i \cdot \ell_{e_i}$. This results in $\alpha' \in [0,1]$, the same return, and the same assets of $v$ as for $\vecAlpha$.

Our result is a reduction to single trades.

\begin{proposition}
    \label{prop:multi-to-single}
    Consider a financial network with monotone payment functions and efficient clearing oracle. For a given multi-trade of incoming edges, there is an additive FPTAS for approximating the optimal improvement of $v$ from any creditor-positive $\vecAlpha$. 
\end{proposition}
\begin{proof}
    Consider a financial network $\calF$ with banks $v$ and $w$ and edges $C$, where $|C|=k$. By Corollary~\ref{cor:multi-to-single}, the multi-trade in $\calF$ can be modeled by a single claims trade with edge $\hat{e}$ in adjusted network $\hat{\calF}$. Invoke the FPTAS to compute a haircut rate $\alpha$ for the single claim in $\hat{\calF}$. This results in assets of $\alpha \cdot \sum_{i=1}^k \ell_{e_i} + \sum_{e \in E'^-(v)} p'_e$ for $v$ in $\hat{\calF}$. Clearly, the same value is obtained with the multi-trade when all haircut rates are set to $\alpha$, i.e., $\alpha_i = \alpha \; \forall i \in [k]$. Clearly, this choice of haircut rates also yields an (approximately) optimal solution for the multi-trade. 
\end{proof}
Combining the insight with Propositions~\ref{prop:single-trade-prop} and~\ref{prop:singleTradeRank}, we obtain the following corollary.
\begin{corollary}
    Consider a financial network with proportional or edge-ranking payments. For a given multi-trade of incoming edges, there are efficient algorithms to compute an optimal creditor-positive $\vecAlpha^*$ or decide that none exists. 
\end{corollary}

\subsection{Choosing Subsets of Claims}
\label{sec:multiIncChoose}

For a fixed pair of creditor $v$ and buyer $w$, the incoming edges of $v$ yield an exponential number of different edge sets $C$ that might be used for a multi-trade. Thus, a creditor-positive multi-trade \emph{cannot} be derived trivially by checking feasibility for all sets $C$. For improving the assets of $v$ by a multi-trade with buyer $w$, the selection of claims to be traded is critical. How can we compute a (near-)optimal set of incoming edges $C \subseteq (E^-(v) \setminus E^+(w))$ for a creditor-positive multi-trade with $w$ such that we maximize the improvement of $v$?

Again, we study the problem when the haircut rates are chosen as part of the trade (see Appendix~\ref{app:fixedAlpha} for the problem when haircut rates are fixed in advance). 
The challenge is to find a set of claims $C$ with creditor $v$ and appropriate individual haircut rates $\alpha_i$, for $e_i \in C$. The resulting multi-trade shall be creditor-positive and yield the maximal improvement for $v$ (over all creditor-positive multi-trades of incoming edges of $v$). 

We show that this problem is \classNP-hard, for every set of monotone payment functions. Formally, we show it is \classNP-hard to decide whether creditor $v$ can be \emph{saved} by a creditor-positive multi-trade of incoming edges, i.e., whether total assets of $L_v$ can be achieved. We call this problem \textsc{IncomingSave-VR} (for variable haircut rates).

In the class of networks we construct for the reduction, every bank has at most one outgoing edge. Hence, all payments will be independent of the payment function that is used. Moreover, once a set of claims $C$ is chosen, finding optimal haircut rates for the multi-trade of $C$ to $w$ is trivial in this class of networks. Hardness arises from the choice of $C$.
\begin{restatable}{theorem}{thmincoming}\label{thm:incoming-var-a-NP}
    \textsc{IncomingSave-VR} is weakly \classNP-hard.  
\end{restatable}

\subsubsection{Approximate Claims Trades}

Contrasting \classNP-hardness, we show that the problem to compute a multi-trade improving $v$ by a given amount can be solved efficiently when slightly relaxing the liability condition. 
\begin{definition}[$\varepsilon$-Approximate Multi-Trade]~\label{def:approxMultiTrade}
    A multi-trade $C$ with creditor $v$, buyer $w$, return $\rho>0$ and haircut rates $\alpha_i \in [0,1+\varepsilon]$, for all $e_i \in C$, is called $\varepsilon$-approximate if $\rho \leq (1+\varepsilon) \cdot \sum_{e_i \in C} \ell_{e_i}$.
\end{definition}
Consider a \emph{creditor-positive $\varepsilon$-approximate} trade. Such a trade (1) strictly increases the assets of $v$ and exactly maintains the ones of $w$, (2) is affordable by $w$, i.e., $\rho = \sum_{e_i \in C} \alpha_i \ell_{e_i} \le a_w^x$, and (3) satisfies exact fixed-point conditions in the emerging clearing state. It is approximate only in the liability condition of the trade.

We construct a bicriteria FPTAS to compute a creditor-positive multi-trade of incoming edges. Suppose $\varepsilon, \delta > 0$ such that $1/\varepsilon$ is polynomial and $1/\delta$ is exponential in the representation size of $\calF$. Our FPTAS guarantees that the computed trade is $\varepsilon$-approximate and yields assets of at least $A^*-\delta$ for $v$, where $A^*$ are the assets of $v$ resulting from an \emph{optimal exact} creditor-positive trade. The FPTAS uses a connection to the \knapsack\ problem. 

We proceed in several steps: First, we consider computing an (exact) trade that achieves a target asset value $A$ for the creditor. For this problem, we derive \knapsack-style constraints capturing a set of three necessary and sufficient conditions of a valid creditor-positive trade. We then adapt the dynamic program for \knapsack\ to construct an FPTAS to compute an $\varepsilon$-approximate multi-trade with assets value at least $A$ for $v$ in polynomial time. Finally, we show how to use binary search to find a trade with asset level at least $A^*-\delta$.

\myparagraph{Necessary and Sufficient Conditions} As a first step, we consider exact trades that achieve assets of at least $A$ for $v$. Suppose there is such a trade with a set $C$ of traded edges, and let $k =  |C|$. Consider $\calF'$ after trade $C$ has occurred with a suitably chosen return. Since $C$ is fixed, by Corollary~\ref{cor:multi-to-single} we can express the outcome of the trade using a single claims trade. Now apply the split network $\calF^{\spl}$ and Lemma~\ref{lem:monoton}. Hence, using
\[
d_v^{\spl} = A - \left(a^x_v + \sum_{e' \in (E^-(v) \setminus C)} p'_{e'}\right) \quad \text{ and }\quad 
d_w^{\spl} = \left(a^x_w + \sum_{e' \in (E^-(w) \cup C)} p'_{e'}\right) - a_w\enspace,
\]
a creditor-positive trade with set $C$ and asset value at least $A$ exists if and only if $d_v^{\spl} = d_w^{\spl} > 0$. This implies that
\begin{equation}
    \label{eq:feasible}
    (A - a_v^x) + (a_w - a_w^x) = \sum_{e \in E^-(v) \cup E^-(w)} p'_e
\end{equation}
must hold. This condition is \emph{independent} of the set $C$ of traded edges. As such \eqref{eq:feasible} is a necessary condition that any creditor-positive trade with asset level at least $A$ for $v$ can exist.

With return $\rho = d_v^{\spl}$ for the given set $C$, we satisfy the fixed-point conditions in $\calF'$. By Lemma~\ref{lem:monoton} the optimal trade using the given set $C$ only yields a larger return, i.e., $\rho \ge d_v^{\spl} > 0$. Moreover, the liabilities of the traded edges must be high enough to allow the return $\rho$, i.e., $\sum_{e_i \in C} \ell_{e_i} \geq \rho$. Using $P' = \sum_{e \in E^-(v)} p'_e$ this necessary condition is expressed by
\begin{equation}\label{eq:value}
    \sum_{e \in C} \ell_e \geq d_v^{\spl} = A - a_v^x - P' + \sum_{e \in C} p'_e \quad \Longleftrightarrow \quad
    \sum_{e \in C} (\ell_e - p'_e) \geq  A - a^x_v - P' \; .
\end{equation}
$w$ must be able to pay the required return. Since the return is solely funded by external assets, we obtain the necessary condition
\begin{equation}\label{eq:weight}
a^x_w \geq \rho = d_v^{\spl} \quad \Longleftrightarrow \quad \sum_{e \in C} p'_e \leq a^x_w - A + a_v^x + P' \; .
\end{equation}
While each condition \eqref{eq:feasible}-\eqref{eq:weight} is necessary, it is easy to see that in combination they are sufficient. Indeed, if they hold, then there is a creditor-positive trade of set $C$ with return $\rho \ge d_v^{\spl} = d_w^{\spl} > 0$ that respects the exact liabilities of traded edges, is affordable by $w$, and achieves asset level at least $A$ for $v$. We summarize the argument in the following lemma:
\begin{lemma}
    \label{lem:knapsack}
    For a given set $C$ of incoming edges of $v$, the following are equivalent:
    \begin{enumerate}
        \item There is a multi-trade of $C$ to $w$ that achieves an asset level at least $A$ for $v$.
        \item Equation \eqref{eq:feasible} holds, and the set $C$ satisfies \eqref{eq:value} and \eqref{eq:weight}.
    \end{enumerate}
\end{lemma}

\myparagraph{Knapsack-Style FPTAS} 
While condition \eqref{eq:feasible} can be checked directly after computing the clearing state $\vecp'$, determining the existence of a set $C$ that satisfies conditions \eqref{eq:value} and \eqref{eq:weight} can be cast as a \knapsack\ decision problem: For each edge $e \in E^-(v)$ the payments $p'_e$ are the non-negative \emph{weight} of $e$, and the residual $\ell_e-p'_e$ is the non-negative \emph{value} of $e$. Decide the existence of a subset of edges with total value lower bounded by \eqref{eq:value} and total weight upper bounded by \eqref{eq:weight}.

We next adapt the standard FPTAS for \knapsack\ to compute an \emph{approximate} multi-trade. We \emph{round up} the residual $\ell_e - p'_e$ of every edge to the next multiple of a parameter $s$. This can be interpreted as increasing the liabilities $\ell_e$ by a small amount. We then determine if \eqref{eq:value} and \eqref{eq:weight} allow a feasible solution by using the standard dynamic program for \knapsack\ in polynomial time. We term this procedure the Level-FPTAS.

\begin{lemma}[Level-FPTAS]
    \label{lem:FPTAS}
    For a given number $A$, suppose there exists a creditor-positive multi-trade of incoming edges that yields assets at least $A$ for $v$. Then, for every $\varepsilon > 0$, there is an algorithm to compute an $\varepsilon$-approximate creditor-positive trade with assets at least $A$ for $v$ in time polynomial in the size of $\calF$ and $1/\varepsilon$.
\end{lemma}
\begin{proof}
    First, check feasibility of condition~\eqref{eq:feasible} since otherwise the desired trade does not exist. Then, consider all incoming edges $e \in E^-(v)$ of $v$ and define $m=|E^-(v)|$. We denote by $r_e = \ell_{e}-p'_{e}$ the residual of $e$. Let $r_{max}$ be the maximal residual capacity with respect to $\vecp'$ of any edge that satisfies the weight constraint, i.e., $r_{max} = \max \{ r_e \mid e \in E^-(v), p'_e \le a_w^x - A + a_v^x + P'\}$. Round the residual capacities \emph{up} using a scaling factor $s = \frac{\varepsilon \cdot r_{max}}{m}$. Then determine the optimal solution $C^*$ for the rounded values $\tilde{r}_e = s \cdot \lceil (r_e)/ s \rceil$ via the standard dynamic program for \knapsack. The running time is bounded by $O(m^3/\varepsilon)$. 

    Using a return of $\rho = A - \left(a_v^x + \sum_{e \in E^-(v) \setminus C^*} p'_e\right)$ the trade of $C^*$ yields a clearing state in $\calF'$ with assets at least $A$ for $v$. By definition, $C^*$ satisfies \eqref{eq:weight}, and since the instance satisfies~\eqref{eq:feasible}, $\rho \le a^x_w$ is also guaranteed. 
    
    Regarding the liabilities, note that
    \begin{equation*}
        \sum_{e \in C^*} \tilde{r}_e \; \leq \; \sum_{e \in C^*} r_e + s 
        \; \leq \; \varepsilon r_{max} + \sum_{e \in C^*} r_e
        \; \leq \; (1+\varepsilon) \sum_{e \in C^*} r_e
        \; \leq \; \sum_{e \in C^*} \ell_e (1+\varepsilon) - p'_e\enspace. 
    \end{equation*}
    There exists an exact trade $C$ with assets at least $A$ for $v$, so the trade with $C$ satisfies \eqref{eq:value} and \eqref{eq:weight}. As such,
    \[
        \sum_{e \in C^*} \tilde{r}_e \quad \ge \quad \sum_{e \in C} \tilde{r}_e \quad \ge \quad \sum_{e \in C} r_e \quad \ge \quad A - a_v^x - P'\enspace,
    \]
    so the optimal solution $C^*$ satisfies \eqref{eq:value} using the rounded residuals. Hence, 
    \begin{align*}
        \rho \quad &= \quad A - \left(a_v^x + \sum_{e \in E^-(v) \setminus C^*} p'_e\right) \quad = \quad A - a_v^x - P' + \sum_{e \in C^*} p'_e \quad \le \quad \sum_{e \in C^*} \tilde{r}_e + p'_e \\ 
        &\le \quad (1+\varepsilon) \sum_{e \in C^*} \ell_e\enspace,
    \end{align*}
    so the return generated by $C^*$ violates the liability condition by at most a factor of $1+\varepsilon$.
\end{proof}

The lemma gives rise to an efficient algorithm computing an $\varepsilon$-approximate multi-trade with assets at least $A$, whenever an exact trade with assets at least $A$ exists. Similar to Theorem~\ref{thm:singleTrade}, we use this test to construct a bicriteria FPTAS.

\myparagraph{Bicriteria-FPTAS}
As our main result for multi-trades of incoming edges, we obtain a bicriteria FPTAS. Suppose assets $A^*$ for $v$ are achievable by an exact creditor-positive multi-trade of incoming edges. We will compute an $\varepsilon$-approximate one resulting in assets at least $A^*-\delta$ for $v$, for any $\varepsilon, \delta > 0$. We say such a trade is \emph{$\delta$-optimal}.

\begin{theorem}\label{thm:approx}
    Consider a financial network with monotone payment functions and efficient clearing oracle, creditor $v$ and buyer $w$. If there exists a creditor-positive multi-trade of incoming edges, then an $\varepsilon$-approximate $\delta$-optimal trade can be computed in time polynomial in the size of $\calF$, $1/\varepsilon$ and $\log 1/\delta$, for every $\varepsilon, \delta > 0$,
\end{theorem}
\begin{proof}
    We use the binary search idea put forward in Theorem~\ref{thm:singleTrade}. We choose $\delta > 0$ and apply binary search over the set $\{a_v + \delta, a_v + 2\delta, \dots, M_v\}$ of potential asset values for $v$. Recall that $M_v$ is an upper bound for maximal achievable assets of $v$. For multi-trades,  $M_v$ is upper bounded by $a^x_v + a^x_w + \sum_{e \in E^-(v)} \ell_e$. The goal is to find an asset value that is as large as possible.
    
    Running the Level-FPTAS with any value from the interval $A' \in (a_v, A^*]$, we are guaranteed to receive an approximate multi-trade with asset level at least $A'$ for $v$ in polynomial time. As such, the binary search will never terminate with a value of $A' \le A^* - \delta$. The search terminates in at most $\lceil \log_2 (1+(M_v - a_v)/\delta)\rceil$ steps. 
\end{proof}

When called with an asset level $A' > A^*$, the Level-FPTAS might or might not return a corresponding multi-trade -- rounding up the residuals can introduce non-monotone behavior. As such, using binary search our algorithm does not necessarily return an optimal asset value of $v$ for any creditor-positive $\varepsilon$-approximate multi-trade. However, since the Level-FPTAS never fails to return a multi-trade for any asset level $A' \le A^*$, we are guaranteed that assets of more than $A^*-\delta$ for $v$ are achieved.

\myparagraph{Ranking Payments}
For edge-ranking payments, the set of meaningful values to be tested for $A^*$ in the binary search can be restricted to a grid of at most exponential precision in the input size. This allows to compute an $\varepsilon$-approximate multi-trade with assets at least $A^*$, i.e., such a trade is $\delta$-optimal with $\delta = 0$.
\begin{corollary}
    \label{cor:tradeRank}
    Consider a financial network with edge-ranking functions, creditor $v$ and buyer $w$. If there exists a creditor-positive multi-trade of incoming edges, then an $\varepsilon$-approximate 0-optimal trade can be computed in time polynomial in the size of $\calF$ and $1/\varepsilon$.
\end{corollary}
\begin{proof}
    Consider an optimal exact multi-trade $C$ with return $\rho$ that achieves optimal assets of $A^*$ for $v$. Recall that all liabilities and external assets are integers, and so is $M_v$. If $A^*$ is integral, then we can run the binary search with $\delta = 1$ and obtain an approximate trade with assets of (more than $A^*-1$ and, thus) at least $A^*$ for $v$.

    To show that $A^*$ is integral, consider an optimal creditor-positive trade $C$ achieving assets $A^*$. We resort to the equivalent representation as a single claims trade (Corollary~\ref{cor:multi-to-single}). For this single trade, consider the return network $\calF^{\ret}$. An optimal return $\rho$ for a trade of edge set $C$ evolves as the payment $p_{e_r}^{\ret}$ on $e_r$ in the clearing state $\vecp^{\ret}$. Recall that the payment function of $w$ in the return network is also an edge-ranking function (c.f.\ Proposition~\ref{prop:singleTradeRank}). For edge-ranking payments, if all liabilities and external assets are integers, then the clearing state has integral payments~\cite{BertschingerStrategicPayments}. The assets of every bank in $\vecp^{\ret}$ (and $A^*$) are integral.
\end{proof}

\section{Multi-Trades of Outgoing Edges}
\label{sec:multiOut}

In this section, we study multi-trades of \emph{outgoing} edges of a bankrupt bank $u$. We strive to improve the assets of $u$'s creditors directly (and not indirectly via $u$ through trades of incoming edges). It might not be feasible to save a particular bank $u$, e.g., when its debt is too high in relation to the claims. In such cases, we attempt to minimize the \emph{contagion} of bankruptcy from $u$ to her creditors by conducting multi-trades of outgoing edges of $u$. We execute multi-trades that \emph{maximize} the total profit of all creditors, not just those involved in the trade. No creditor nor buyer $w$ should be harmed by the trade.

\begin{definition}[Pareto-positive trade] \label{def:outgoing}
    Let $v_1, v_2, \dots, v_l$ be $u$'s creditors.
    A multi-trade of outgoing edges of $u$ to $w$ is called Pareto-positive, if $a'_{v_i} > a_{v_i}$ for at least one creditor $v_i$, $a'_{v_i} \geq a_{v_i}$ for all creditors and $a'_w \geq a_w$.  
\end{definition}

Suppose we are given a financial network \F with banks $u,w$ and a set $C$ of $k$ outgoing edges of $u$. Denote the creditors of edges $C$ by $V_C=\{v_i\mid e_i \in C , \cre(e_i)=v_i\}$. A collection of haircut rates $\vecAlpha=(\alpha_1,\alpha_2,\dots, \alpha_k)$ is called \emph{Pareto-positive} if $C$ together with $\vecAlpha$ forms a Pareto-positive multi-trade. The objective is to derive the optimal values of $\vecAlpha$ which maximize the sum of profit of creditors $v_1,v_2,\dots,v_l$, i.e., $\max \sum_{i=1}^l a'_{v_i} - a_{v_i}$. Appendix~\ref{app:fixedAlpha} contains a discussion for the scenario where $\vecAlpha$ is fixed and part of the input.

Consider the problem where set $C$ is not given as part of the input but is chosen as part of the solution. The goal is to select a subset of $u$'s outgoing edges $C \subseteq E^+(v)$ together with a vector of haircut rates $\vecAlpha = (\alpha_1,\alpha_2,\dots, \alpha_{|C|})$ such that the multi-trade is Pareto-positive and maximizes the improvement of $u$'s creditors, i.e., $\sum_{i=1}^l a'_{v_i}- a_{v_i}$. 

In the previous section, we obtained a bicriteria FPTAS for this problem when we trade incoming edges of an insolvent bank. Interestingly, the results hold for all monotone payment functions for which there is an efficient clearing oracle (e.g., edge-ranking or proportional payments). Our results here show a strong contrast -- depending on the payment functions trading outgoing edges can be much harder. For edge-ranking functions (and variable haircut rates), we denote the problem by \MaxER. 
\begin{restatable}{theorem}{thmsetpacking}\label{thm:setpacking}
    \MaxER\ is strongly \classNP-hard. For any constant $\varepsilon > 0$ there exists no efficient $n^{1/2-\varepsilon}$-approximation algorithm for \MaxER\ unless \classP\ = \classNP.
\end{restatable}
Now suppose the set of traded edges $C$ is given as part of the input. Interestingly, the hardness for edge-ranking payments continues to apply when $C$ is fixed a priori. 
\begin{restatable}{corollary}{coroutgoingEdges}\label{cor:outgoingEdges}
    Consider a financial network with banks $u, w$, a set of outgoing edges $C$ of $u$ and edge-ranking payment rules. It is strongly \classNP-hard to determine Pareto-positive haircut rates that maximize the sum of profits of $u$'s creditors. For any constant $\varepsilon > 0$, there exists no efficient $n^{1/2-\varepsilon}$-approximation algorithm unless \classNP\ = \classP.   
\end{restatable}
Finally, we briefly observe that these problems for outgoing edges depend crucially on the set of payment functions. For proportional payments (and variable $\vecAlpha$), the problem for a given set $C$ can be solved efficiently (even if the set $C$ of edges involves different debtors). When the set $C$ of outgoing edges is chosen as part of the solution, we refer to the problem as \MaxProp\ and obtain \classNP-hardness. The approximability status of these problems for different payment functions is an interesting direction for future work.
\begin{restatable}{proposition}{propout}
    For a given financial network with proportional payments and a set $C$ of $k$ edges, there exists an efficient algorithm that computes an optimal Pareto-positive $\vecAlpha^* \in [0,1]^k$ or decides that none exists.
\end{restatable}
\begin{restatable}{theorem}{thmoutgoingProp}\label{thm:outgoingProp}
    \MaxProp\ is strongly \classNP-hard.
\end{restatable}



\bibliography{literature}

\appendix
\clearpage

\section{Missing Proofs}\label{app:proofs}

\propLP*

\begin{proof}
    We derive the desired value for $\alpha$ using a linear program constructed as follows. We use $E'$, $L'$ and $r'$ to refer to the edge sets, total liabilities and recovery rates after the trade, respectively.
      
    For a creditor-positive trade we need to strictly raise the recovery rate $r_v$ and keep $r_w$ unchanged, i.e., $r'_v > r_v$ and $r'_w = r_w$. We relax the condition for $v$ to $r'_v \geq r_v$, since we will discover a solution with strict improvement whenever it exists. Moreover, the return of $w$ is bounded by the external assets of $w$ and the liability of $e$ (since $\alpha \le 1$), i.e., $\rho \leq \min\{ a^x_w, \ell_e\}$.
      
    The clearing state is governed by fixed-point constraints. We reformulate them using the recovery rates. The recovery rate of any bank $b \in (V \setminus \{v,w\})$ must satisfy $r'_b \leq (a^x_b + \sum_{e' \in E'^-(b)} r'_{\deb(e')} \cdot \ell_{e'})/L'_b$. Accounting for the trade, the recovery rate of $v$ is bounded by $r'_v \leq (a^x_v + \rho + \sum_{e' \in E'^-(v)} r'_{\deb(e')} \cdot \ell_{e'})/L'_v$ and similarly $r'_w \leq (a^x_w - \rho + \sum_{e' \in E'^-(w)} r'_{\deb(e')} \cdot \ell_{e'})/L'_w$. 
      
    Finally, as discussed above, we capture the objective by maximizing the total payments in the network after the trade $\max \sum_{e' \in E'} r'_{\deb(e')}\cdot \ell_{e'}$. 

    Overall, the following LP captures the problem of finding an optimal return $\rho$.
		\begin{align*}
				\text{Max.\ } \; &\displaystyle \sum\limits_{e' \in E'} r'_{\deb(e')}\cdot \ell_{e'} &\\
                \text{s.t.\ } \; &\displaystyle r'_v \geq r_v, & \\
                & r'_w = r_w, & \\
                & \rho \leq a^x_w , \\
                & \rho \le \ell_e, \\ 
                & r'_v \leq (a^x_v + \rho + \sum_{e' \in E'^-(v)} r'_{\deb(e')} \cdot \ell_{e'})/L'_v, \\
                & r'_w \leq (a^x_w - \rho + \sum_{e' \in E'^-(w)} r'_{\deb(e')} \cdot \ell_{e'})/L'_w, \\
                & r'_b \leq (a^x_b + \sum_{e' \in E'^-(b)} r'_{\deb(e')} \cdot \ell_{e'})/L'_b & \forall b \in (V \setminus \{v,w\}) ,\\
                & r'_b \in [0,1] & \forall b \in V
		\end{align*}
    The LP is polynomial in the size of $\calF$. An optimal creditor-positive $\alpha$ exists iff the optimal LP solution yields $\rho > p_e$. Then the optimal $\alpha = \rho/\ell_e$.
\end{proof}


\thmincoming*

\begin{proof}
    We show the statement with reduction from \subsetSum. An instance of \subsetSum\ is given by a multiset $S=\{b_1,b_2,\dots,b_n\}$ of positive integers, i.e., $b_i \in \N_{>0}$, and a target value $T \in \N$ with $T \leq \sum_{i=1}^k b_i$. The goal is to decide whether there exists a subset of $S' \subseteq [n]$ with sum of elements exactly $T$, i.e., $\sum_{i \in S'}b_i = T$.

    The graph $G$ consists of $n+2$ banks where $V= \{v,w,b_1,b_2,\dots,b_n\}$. We introduce edge $e_i$ with $\cre(e_i)=v$ and $\deb(e_i)=b_i$ and liability $2b_i$, for all $i \in [n]$. Additionally include edge $e$ with debtor $v$ and creditor $w$ with liability $\ell_e = T+\sum_{i=1}^n b_i$. Note that we use $b_i$ both to refer to bank $b_i$ as well as to integer $b_i$. Lastly, set the external assets of $w$ to $2T$. For every bank $b_i$, the external assets are the integer $b_i$. Observe that bank $b_i$ pays $b_i$ to $v$. Hence, $v$ receives incoming payments of $\sum_{i=1}^n b_i$ and forwards the same amount to $w$.

    We claim that there exists a subset $S'\subseteq [n]$ with sum of elements exactly $T$ if and only if there exists a creditor-positive multi-trade of incoming edges and individual haircut rates that saves $v$.

    Assume there exists a subset $S'\subseteq [n]$ with sum of elements exactly $T$. Then, a creditor-positive multi-trade can be constructed as follows. For every $i \in S'$, add $e_i$ to $C$ and assign $\alpha_i = 1$. Thus, the return is given by $\rho=\sum_{i \in S'}\alpha_i \cdot \ell_{e_i}=\sum_{i \in S'} 2b_i = 2T \leq a^x_w$, for $|C|=k$. Because the payments of the remaining incoming edges of $v$ are not affected by the multi-trade, the total assets of $v$ equal
    \[
        a'_v = \sum_{i \in S'} \alpha_i \cdot \ell_{e_i} + \sum_{e \in E'^-(v)}p'_e = 2T + \sum_{e\in E^-(v)} p_e - \sum_{i\in S'}p_{e_i}=2T + a_v - T = a_v+T \; .
    \]
    Consequently, $v$'s total assets are increased by $T$, pays all debt, all these payments reach $w$, leaving $w$ indifferent.

    For the other direction, assume there exists a creditor-positive multi-trade given by $C\subseteq E^-(v)$, with $|C|=k$, and an $\alpha_i$, for every $e_i\in C$. Furthermore, assume $v$ improves by $T$. We show that $S' = \{i \mid e_i \in C\}$ is a solution for \subsetSum. Again, since no cycles can evolve, $v$ has to receive additional funds via the return $\rho$. It can be decomposed into two parts: First, $\sum_{e_i \in C} p_{e_i}$ can be interpreted as reimbursement for the traded edges, whose payments $v$ no longer receives; Second, $\rho - \sum_{e_i \in C}p_{e_i}$ is the surplus that helps to make $v$ solvent. Since the payment on edges ${E'}^-(v)$ are unaffected by the multi-trade, $\rho-\sum_{e_i \in C}p_{e_i}$ is exactly the improvement of $v$. Since $v$ becomes solvent $\rho-\sum_{e_i\in C}p_{e_i}=T$, and since $\rho\leq \sum_{e_i\in C}\ell_{e_i}$, this implies $\rho - \sum_{e_i\in C}p_{e_i} \leq \sum_{e_i \in C} \ell_{e_i}-\sum_{e_i \in C}p_{e_i}$. Intuitively, the surplus is bounded by the sum of residual debt on the traded edges. $w$ pays surplus of $T$ and has external assets of $2T$, so $\rho\leq a^x_w = 2T$ implies $\sum_{i\in S'}b_i=\sum_{e_i \in C}p_{e_i}\leq T$. Assume $\sum_{e_i \in C}p_{e_i}<T$, i.e., $w$ buys claims with debtors $b_i$ such that $\sum_{i\in S'}b_i <T$. Then, by construction, the residual debt of the traded edges sums up to less than $T$ and, thus, do not allow for a surplus of $T$. Thus, $\sum_{i \in S'}b_i \geq T$ and, hence, $\sum_{i\in S'}b_i = T$. Thus, $S'$ is a solution to the instance of \subsetSum. 
\end{proof}


\thmsetpacking*

\begin{proof}
    \begin{figure}
    \centering
    \begin{tikzpicture}[>=stealth', shorten >=1pt, auto, node distance=1.5cm, transform shape, align=center]
        \node[circle, draw, inner sep=3.3] (v) at (0,0) {$v$};
        \node[bank] (s1) [below left=1cm and 3.3cm of v] {$S_1$};
        \node[bank] (s2) [right=2cm of s1] {$S_2$};
        \node[bank] (s3) [right=of s2] {$S_3$};
        \node (sdots) [right=0.4cm of s3] {$\cdots$};
        \node[bank] (sl) [right=2cm of s3] {$S_l$};
        \node[circle, draw, inner sep=2] (u3) [below=2.42cm of v] {$u_3$};
        \node[circle, draw, inner sep=2] (u2) [left=of u3] {$u_2$};
        \node[circle, draw, inner sep=2] (u1) [left=of u2] {$u_1$};
        \node[circle, draw, inner sep=2] (u4) [right=of u3] {$u_4$};
        \node (edots) [right=0.4cm of u4] {$\cdots$};
        \node[circle, draw, inner sep=1.8] (um) [right=of u4] {$u_m$};
        \node[circle, draw, inner sep=2.8] (w) [below=1cm of u3] {$w$};
        \node[rectangle, draw=black, inner sep=2.5pt] (xw) [below=0.1cm of w] {\footnotesize $k(M+m)$};
        \draw[->] (v) --node[left=0.3cm, pos=0.4] {\footnotesize$M+d_1$} (s1);
        \draw[->] (v) --node[left=0.05cm, pos=0.7] {\footnotesize$M+d_2$} (s2);
        \draw[->] (v) --node[left=0.1cm, pos=0.7] {\footnotesize$M+d_3$} (s3);
        \draw[->] (v) --node[right=0.3cm, pos=0.4] {\footnotesize$M+d_l$} (sl);
        \draw[->] (s1) --node[left, pos=0.3] {\footnotesize{$1$}} (u1);
        \draw[->] (s1) --node[left, pos=0.3] {\footnotesize{$1$}} (u2);
        \draw[->] (s1) --node[right=0.2cm, pos=0.1] {\footnotesize{$1$}} (u3);
        \draw[->] (s2) --node[right=0.2cm, pos=0.25] {\footnotesize{$1$}} (u1);
        \draw[->] (s2) --node[right, pos=0.4] {\footnotesize{$1$}} (u3);
        \draw[->] (s2) --node[right=0.2cm, pos=0.1] {\footnotesize{$1$}} (u4);
        \draw[->] (s3) --node[left, pos=0.2] {\footnotesize{$1$}} (u3);
        \draw[->] (s3) --node[left, pos=0.4] {\footnotesize{$1$}} (u4);
        \draw[->] (sl) --node[left, pos=0.4] {\footnotesize{$1$}} (um);
        \draw[dashed, ->] (s1) to[in=200, out=200, looseness=2] node[left=0.1cm, pos=0.5] {\footnotesize{$M$}} (w);
        \draw[dashed, ->] (s2) --node[left, pos=0.5] {\footnotesize{$M$}} (w);
        \draw[dashed, ->] (s3) --node[right, pos=0.5] {\footnotesize{$M$}} (w);
        \draw[dashed, ->] (sl) to[in=340, out=340, looseness=2] node[right=0.1cm, pos=0.5] {\footnotesize{$M$}} (w);
        \draw[->] (u1) --node[below,pos=0.4] {\footnotesize $1$}(w);
        \draw[->] (u2) --node[above,pos=0.5] {\footnotesize $1$}(w);
        \draw[->] (u3) --node[left=-0-.05cm,pos=0.3] {\footnotesize $1$}(w);
        \draw[->] (u4) --node[above,pos=0.5] {\footnotesize $1$}(w);
        \draw[->] (um) --node[below,pos=0.4] {\footnotesize $1$}(w);
    \end{tikzpicture}
    \caption{Schematic visualisation of the constructed network in proof of Theorem~\ref{thm:setpacking}. The dashed edges are paid off last with respect to the edge-ranking rules.}
    \label{fig:setpacking}
    \end{figure}
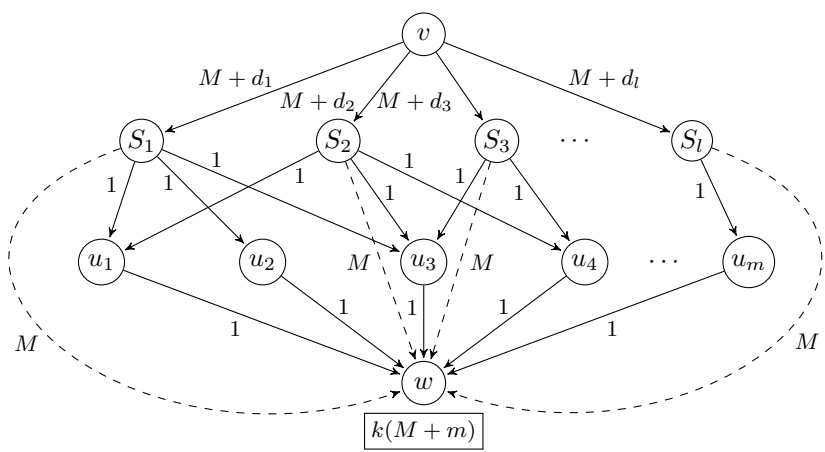

    We first show \classNP-hardness and then observe inapproximability as a consequence.

    \myparagraph{\classNP-hardness}
    We perform a reduction from \setPacking. In such an instance, we are given a graph universe $U = \{u_1,u_2, \dots, u_m\}$, a family $\mathcal{S}=(S_i)_{i \in [l]}$ of sets $S_i \subset U$. We denote the cardinality of set $S_i$ by $d_i = |S_i|$. For a given value $k \leq l$, the goal is to decide whether there exists a set packing $P \subseteq [l]$ of $k$ pair-wise disjoint sets, i.e., $S_i \cap S_j = \emptyset$ for all $i,j \in P$ and $|P|=k$. 

    We construct a financial network $\calF$ with banks $V=\{v,w,S_1,S_2,\dots,S_l,u_1,u_2,\dots,u_m\}$ as shown in Figure~\ref{fig:setpacking}. Note that we use $S_i$ and $u_j$ to both refer to the set or element, respectively, as well as the corresponding node in the financial network.
    For every set $S_i$, add an edge $e_i$ with $\deb(e_i)=v$, $\cre(e_i)=S_i$ and liability $\ell_{e_i}=M+d_i$, where $M = m^3$. Furthermore, include an edge $e$ with $\deb(e)=S_i$ and $\cre(e)=u_j$ if $u_j \in S_i$ and set the liability to $\ell_e = 1$. Moreover, every bank $S_i$ has an additional outgoing edge $e$ with liability $M$ and $\cre(e)=w$. Finally, add an edge $e_j$ with $\deb(e_j)=u_j$ and $\cre(e_j)=w$ and liability 1.
    $w$ is the only bank with non-zero external assets $k(M+m)$.

    We assume $v$ and banks $u_j$ have arbitrary edge-ranking payment functions. Since $v$ has no external and no incoming assets, she will never make any positive payments. Banks $u_j$ each have a unique outgoing edge and thus will direct all assets towards this edge.
    Each bank $S_i$ has an arbitrary edge-ranking rule with the condition that edge $e$ with creditor $w$ is paid last. In other words, $S_i$ will only pay positive assets to $e$ if her total assets are sufficient to saturate all other outgoing edges.  

    We claim that there exists a set packing of size $k$ if and only if there is a multi-trade of outgoing edges such that the creditors of $v$ profit by a total amount of at least $k(M+m)$.

    First assume that there exists a set packing $P$ of size $k$. Then, perform a multi-trade of edges $e_i$, where $\deb(e_i)=v, \cre(e_i)=S_i$ and $S_i \in P$, and set all haircut rates to 1. Hence, the trade consists of $k$ edges and $w$ pays a return of $M+d_i$ to every bank $S_i$. Clearly, all banks $S_i$ are able to saturate all outgoing edges. Since $P$ is a set packing, every bank $u_j$ receives incoming payments of at most $1$ and forwards the assets to $w$. Consequently, the incoming payments of buyer $w$ equal the total return of $kM + \sum_{S_i \in P} d_i$. Hence, the multi-trade is Pareto-positive and the creditors of $v$ profit by a total amount of $kM + \sum_{S_i \in P} d_i$.

    For the other direction, assume there exists a Pareto-positive multi-trade where $v$'s creditors profit by a total amount of at least $kM$, i.e., $\sum_{i=1}^{l} a'_{S_i} \geq kM + \sum_{i=1}^{l}a_{S_i}$.
    Observe that no bank $S_i$ can receive positive incoming payments from $v$ for every multi-trade of outgoing edges and, thus, all assets of $S_i$ must stem from the received return. Simultaneously, since $w$ is not harmed by the multi-trade, she must receive the full return as incoming payments. As $k \le n$ and $M = m^3$, this is only possible if there are positive payments on at least $k$ edges with liability $M$. This directly implies that at least $k$ banks $S_i$ settled all debt towards all their creditors $u_j$ and all assets were forwarded to $w$. Consequently, all $k$ banks have liabilities towards pair-wise disjoint subsets of $U$. 

    \myparagraph{Inapproximability}
    For contradiction, we assume a $n^{1/2-\varepsilon}$-approximation algorithm for the problem exists. We use this algorithm to construct an efficient approximation algorithm for \independentSet\ with a sublinear approximation ratio, which implies \classP\ = \classNP~\cite{Zuckerman07}.
    
    An instance of \independentSet\ is given by an undirected graph and an integer $k$. We resort to the standard interpretation as an instance of \setPacking, i.e., the universe is the set of edges, and each vertex corresponds to a set of incident edges. Note that this implies $l \le 2m$ and $m < l^2$. Now construct the corresponding financial network for \MaxER\ as described in the reduction of Theorem~\ref{thm:setpacking}. Note that the network has a number of nodes that is $n = l + m + 2 \le 3m+2$. With foresight, we set $M$ in the reduction to $M = (3m+2)^3 \ge n^3$. Observe that this does not change the correctness of the reduction.
    
    Suppose the optimal solution for the constructed network is given by a multi-trade with a total return of $k^*M + \beta^*$. By assumption, the approximation algorithm computes a solution with total return at least $kM+\beta$ satisfying $n^{1/2-\varepsilon} \cdot (kM+\beta) \geq k^*M + \beta^*$. We know that $1 \le k,k^* \le m < n$, so $1 \le \beta, \beta^* \le m^2 < n^2$. Since $M = n^3$, it holds that $kM+\beta \leq (k+1/n)M$. Combining these observations yields $n^{1/2-\varepsilon} (k + 1/n) M \ge k^*M$ and hence $n^{1/2-\varepsilon} > k^*/(k+1/n)$.

    If the optimal solution for the instance of \MaxER\ achieves $k^*M+\beta^*$, clearly there exists a set packing of size $k^*$ but not of size $k^*+1$. Hence, $k^*$ is the maximal cardinality of a feasible set packing in the given \setPacking\ instance. Moreover, with the reduction to \MaxER\ we obtain a feasible set packing of at least size $k$. In other words, we receive a $k^*/k$ approximation for \setPacking\ and, thus, for \independentSet. This results in a approximation factor of
    \[
        \frac{k^*}{k} = \frac{k^*}{k+1/n} \cdot \frac{k+1/n}{k} \leq n^{1/2-\varepsilon} \cdot \frac{k+1/n}{k} < n^{1/2-\varepsilon} + 1\;.
    \]
    Now observe that $n^{1/2-\varepsilon} + 1 \le (3m+2)^{1/2-\varepsilon} + 1 < (3 l + 2)^{1-\varepsilon} + 1 < l^{1-\delta}$ for a suitably small constant $\delta = \delta(\varepsilon) > 0$ and any $l$ larger than a suitable constant $l_0(\varepsilon)$. Recall that $l$ is the number of nodes in the underlying instance of \independentSet. Hence, we obtain an algorithm with a sublinear approximation ratio for \independentSet.
\end{proof}


\coroutgoingEdges*

\begin{proof}
    To observe this, consider the constructed network as described in proof of Theorem~\ref{thm:setpacking} and add an additional edge from every bank $S_i$ to an auxiliary node each with liability 1. Moreover, include an edge from the auxiliary node to $w$ and set the liability to $l$. Adjust the edge-ranking rules of all $S_i$ so that the new edge has highest priority and increase $w$'s external assets by $l$. Raise the liability of all outgoing edges of $v$ from $M+d$ to $M+d+1$. Finally, fix $C=E^+(v)$. Recall that initially no bank receives any incoming payments. A trivial solution would be to set $\alpha_i = 0$, for all $i \in [l]$, which implies a return of 0 for every creditor. In order to maximize the sum of returns, it is optimal to increase the haircut rate to 1 for the maximal set of edges $C'\subseteq C$ where all associated banks $S_i$ have disjoint creditors.
\end{proof}


\propout*

\begin{proof}
    We compute the optimal solution for $\vecAlpha$ via an LP, if it exists. Recall that the set of edges, recovery rates and total assets after the multi-trade are denoted by $E', L'$ and $r'$, respectively.

    For a Pareto-positive trade, none of $u$'s creditors should be harmed while at least one creditor strictly profits. Formally, we have $r'_{v_i} \geq r_{v_i}$, for all $i \in [l]$, and $\sum_{i:v_i \in V_C} \rho_i + \sum_{i=1}^l (a^x_{v_i} + \sum_{e' \in E'^-(v_i)}r'_{\deb(e')}\cdot l_{e'})  > \sum_{i=1}^l a_{v_i}$. Similarly, $w$ should be indifferent after the trade, i.e., $r'_w = r_w$. Moreover, all returns are upper bounded by the respective liabilities and $w$'s external assets must suffice to cover the sum of returns, i.e., $\rho_i \leq \ell_{e_i}$, for all $e_i \in C$, and $\sum_{i:e_i\in C} \rho_i \leq a^x_w$.

    The recovery rates $r'$ of all banks must be consistent with the clearing state which arises in the network after the trade. For all creditors we have $r'_{v_i} \leq (a^x_{v_i} + \rho_i + \sum_{e' \in E'^-(v_i)} r'_{\deb(e')} \cdot \ell_{e'})/L'_{v_i}$, where $V_C$ is the set of creditors incident to the edges in $C$. Similarly, $r'_w \leq (a^x_w - \sum_{i: e_i \in C}\rho_i + \sum_{e' \in E'^-(w)} r'_{\deb(e')} \cdot \ell_{e'})/L'_w$ must be fulfilled for $w$. For all other banks it must be true that $r'_b \leq (a^x_b + \sum_{e' \in E'^-(b)} r'_{\deb(e')} \cdot \ell_{e'})/L'_b$.

    The objective is to maximize the sum of profits of all creditors $v_1,\dots v_l$. Every creditor $v_i$ can profit in two ways: (i) by the return, in case $e_i \in C$ and the return is chosen as $\rho_i>0$, and (ii) by increased incoming payments in the resulting clearing state. Hence, we set the objective function to 
    \[
        \sum\limits_{i:v_i \in V_C} \rho_i + a_{v_i}^x + \sum_{e' \in E'^-(v_i)} r'_{\deb(e')}\cdot \ell_{e'}\enspace.
    \]
    Observe that maximizing the objective function will also maximize the incoming payments of $u$'s creditors. However, the LP might not derive the correct clearing state for the network after the trade. The clearing state can be determined efficiently upon solving the LP.

    Combining all constraints yields the following LP.

    \begin{align*}
				\text{Max.\ } \; &\displaystyle \sum\limits_{i:v_i \in V_C} \rho_i + a_{v_i}^x + \sum_{e' \in E'^-(v_i)} r'_{\deb(e')}\cdot \ell_{e'} &\\
                \text{s.t.\ } \; &\displaystyle  \sum_{i:v_i \in V_C} \rho_i + \sum_{i=1}^l (a^x_{v_i} + \sum_{e' \in E'^-(v_i)}r'_{\deb(e')}\cdot l_{e'})  \geq \sum_{i=1}^l a_{v_i} , &\\
                & r'_{v_i} \geq r_{v_i} & \forall i \in [l] ,\\
                & r'_w = r_w , & \\
                & \sum_{i: e_i \in C} \rho_i \leq a^x_w , \\
                & \rho_i \le \ell_{e_i}, & \forall e_i \in C ,\\ 
                & r'_{v_i} \leq (a^x_{v_i} + \rho_i + \sum_{e' \in E'^-(v_i)} r'_{\deb(e')} \cdot \ell_{e'})/L'_{v_i}, & \forall v_i \in V_C , \\
                & r'_w \leq (a^x_w - \sum_{i: e_i \in C}\rho_i + \sum_{e' \in E'^-(w)} r'_{\deb(e')} \cdot \ell_{e'})/L'_w, \\
                & r'_b \leq (a^x_b + \sum_{e' \in E'^-(b)} r'_{\deb(e')} \cdot \ell_{e'})/L'_b & \forall b \in (V \setminus (\{w\}\cup V_C)) ,\\
                & r'_b \in [0,1] & \forall b \in V
		\end{align*}

  The LP has polynomial size in \F. An Pareto-positive multi-trade exists if $$\sum_{i:v_i \in V_C} \rho_i + \sum_{i=1}^l (a^x_{v_i} + \sum_{e' \in E'^-(v_i)}r'_{\deb(e')}\cdot l_{e'})  > \sum_{i=1}^l a_{v_i}$$ is satisfied. The optimal haircut rates are defined by $\alpha_i = \rho_i/\ell_{e_i}$, for all $e_i \in C$.
\end{proof}


\thmoutgoingProp*

\begin{proof}
    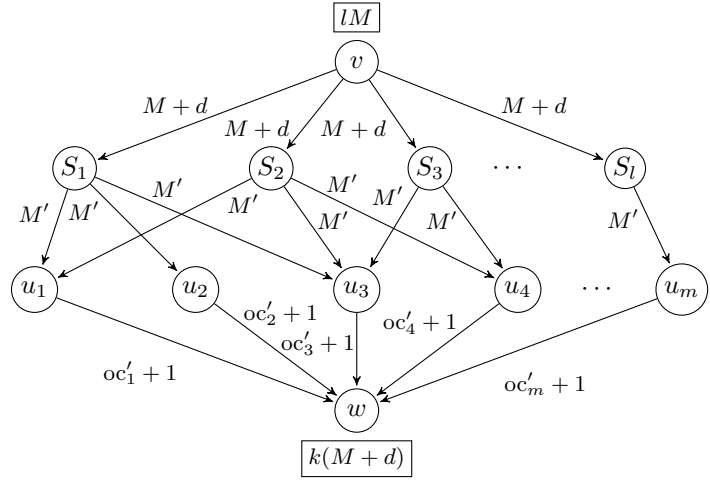
\begin{figure}
    \centering
    \begin{tikzpicture}[>=stealth', shorten >=1pt, auto, node distance=1.5cm, transform shape, align=center]
        \node[circle, draw, inner sep=3.3] (v) at (0,0) {$v$};
        \node[bank] (s1) [below left=1cm and 3.3cm of v] {$S_1$};
        \node[bank] (s2) [right=2cm of s1] {$S_2$};
        \node[bank] (s3) [right=of s2] {$S_3$};
        \node (sdots) [right=0.4cm of s3] {$\cdots$};
        \node[bank] (sl) [right=2cm of s3] {$S_l$};
        \node[circle, draw, inner sep=2] (u3) [below=2.42cm of v] {$u_3$};
        \node[circle, draw, inner sep=2] (u2) [left=of u3] {$u_2$};
        \node[circle, draw, inner sep=2] (u1) [left=of u2] {$u_1$};
        \node[circle, draw, inner sep=2] (u4) [right=of u3] {$u_4$};
        \node (edots) [right=0.4cm of u4] {$\cdots$};
        \node[circle, draw, inner sep=1.8] (um) [right=of u4] {$u_m$};
        \node[circle, draw, inner sep=2.8] (w) [below=1cm of u3] {$w$};
        \node[rectangle, draw=black, inner sep=2.5pt] (xv) [above=0.1cm of v] {\footnotesize $l M$};
        \node[rectangle, draw=black, inner sep=2.5pt] (xw) [below=0.1cm of w] {\footnotesize $k(M+d)$};
        \draw[->] (v) --node[left=0.3cm, pos=0.4] {\footnotesize$M+d$} (s1);
        \draw[->] (v) --node[left=0.05cm, pos=0.7] {\footnotesize$M+d$} (s2);
        \draw[->] (v) --node[left=0.1cm, pos=0.7] {\footnotesize$M+d$} (s3);
        \draw[->] (v) --node[right=0.3cm, pos=0.4] {\footnotesize$M+d$} (sl);
        \draw[->] (s1) --node[left, pos=0.3] {\footnotesize{$M'$}} (u1);
        \draw[->] (s1) --node[left=0.1cm, pos=0.3] {\footnotesize{$M'$}} (u2);
        \draw[->] (s1) --node[right=0.3cm, pos=0.1] {\footnotesize{$M'$}} (u3);
        \draw[->] (s2) --node[right=0.2cm, pos=0.25] {\footnotesize{$M'$}} (u1);
        \draw[->] (s2) --node[right, pos=0.4] {\footnotesize{$M'$}} (u3);
        \draw[->] (s2) --node[right=0.2cm, pos=0.05] {\footnotesize{$M'$}} (u4);
        \draw[->] (s3) --node[left, pos=0.1] {\footnotesize{$M'$}} (u3);
        \draw[->] (s3) --node[left, pos=0.4] {\footnotesize{$M'$}} (u4);
        \draw[->] (sl) --node[left, pos=0.4] {\footnotesize{$M'$}} (um);
        \draw[->] (u1) --node[below=0.3cm,pos=0.3] {\footnotesize $\oc'_1 +1$}(w);
        \draw[->] (u2) --node[right=0.1cm,pos=0.1] {\footnotesize $\oc'_2 +1$}(w);
        \draw[->] (u3) --node[left=-0.1cm,pos=0.4] {\footnotesize $\oc_3'+1$} (w);
        \draw[->] (u4) --node[left=0.1cm,pos=0.2] {\footnotesize $\oc_4'+1$}(w);
        \draw[->] (um) --node[below=0.3cm,pos=0.4] {\footnotesize $\oc'_m +1$}(w);
    \end{tikzpicture}
    \caption{Schematic visualisation of the constructed network in proof of Theorem~\ref{thm:outgoingProp}. Black edge labels indicate liabilities. Due to lack of space we define $M':=M/d+1$ and $\oc'_j:=\oc_j\cdot M/d+1$.}
    \label{fig:setpackingProp}
    \end{figure}

    We prove \classNP-hardness with a reduction from $d$-\setPacking. Every instance consists of a universe $U = \{u_1,u_2, \dots, u_m\}$ and a family $\mathcal{S}=(S_i)_{i \in [l]}$ of sets $S_i \subset U$. All sets have the same cardinality $d$, i.e., $|S_i|=d$ for all $i \in [l]$. For a given value $k \leq l$, the goal is to decide whether there exists a set packing $P \subseteq [l]$ of $k$ pair-wise disjoint sets, i.e., $S_i \cap S_j = \emptyset$ for all $i,j \in P$ and $|P|=k$. Let $\oc_j$ denote the number of occurrences of element $u_j$ in sets $S_i$, i.e., $\oc_j = |\{i\mid u_j \in S_i\}|$.

    We construct a financial network $\calF$ with banks $V=\{v,w,S_1,S_2,\dots,S_l,u_1,u_2,\dots,u_m\}$ as shown in Figure~\ref{fig:setpackingProp}. Note that we use $S_i$ and $u_j$ to both refer to the set or element, respectively, as well as the corresponding node in the financial network.
    For every set $S_i$, add an edge $e_i$ with $\deb(e_i)=v$, $\cre(e_i)=S_i$ and liability $\ell_{e_i}=M+d$, where $M \gg m$. Furthermore, include an edge $e$ with $\deb(e)=S_i$ and $\cre(e)=u_j$ if $u_j \in S_i$ and set the liability to $\ell_e = M/d+1$. Finally, add an edge $e_j$ with $\deb(e_j)=u_j$ and $\cre(e_j)=w$ and liability $\oc_j\cdot M/d +1$. $v$ and $w$ are the only bank with non-zero external assets, where $a^x_v = l\cdot M$ and $a^x_w=k(M+d)$. All banks distribute their assets proportionally which gives rise to the following clearing state. $v$ pays $M$ to every bank $S_i$ who pass $M/d$ to their creditors, respectively. This results in incoming payments of $\oc_j \cdot M/d$ for every $u_j$ and, hence, incoming payments of $\sum_{j=1}^m \oc_j \cdot M/d$ for $w$. 

    We claim that there exists a set packing of size $k$ if and only if there is a Pareto-positive multi-trade of outgoing edges such that the creditors of $v$ profit by a total amount of at least $k\cdot d$.

    First, assume that there exists a set packing $P$ of size $k$. Then, perform a multi-trade of edges $e_i$, where $\deb(e_i)=v, \cre(e_i)=S_i$ and $i \in P$, and set all haircut rates to 1. Hence, the trade consists of $k$ edges and $w$ pays a return of $M+d$ to every bank $S_i$. Clearly, all $k$ banks $S_i$ are able to saturate all outgoing edges. Since $P$ is a set packing, every bank $u_j$ receives incoming payments of at most $\oc_j \cdot M/d +1$ and forwards the assets to $w$. Consequently, the incoming payments of buyer $w$ equal the total return of $k(M+d)$. Hence, the multi-trade is Pareto-positive and the creditors of $v$ profit by a total amount of $k\cdot d$.

    For the other direction, assume there exists a Pareto-positive multi-trade where $v$'s creditors profit by a total amount of at least $k\cdot d$, i.e., $\sum_{i=1}^{l} a'_{S_i} \geq k\cdot d + \sum_{i=1}^{l}a_{S_i}$. Observe that no bank $S_i$ can receive increased incoming payments from $v$ for every multi-trade of outgoing edges and, thus, all additional assets of $S_i$ must stem from the received return. Consider a bank $S_i$ whose incoming edge was traded. Before the trade, $S_i$ receives payments of $M$, hence, $\rho_i \geq M$. The return is upper bounded by $\ell_{e_i}=M+d$, thus, bank $S_i$ profits by at most $d$. By definition of $M$, $w$ is only able to refund payments of at most $k$ edges. Hence, the total profit of $k\cdot d$ is only possible, if $w$ buys $k$ edges with haircut rates 1.
    Simultaneously, since $w$ is not harmed by the multi-trade, she must receive the full return as incoming payments. This implies that all banks $S_i$ with $e_i \in C$ pass additional assets of 1 to each creditor $u_j$ who forward the payments to $w$. Due to the limited liabilities of $u_j$ towards $w$, all $S_i$ must have pair-wise disjoint creditors. In conclusion, there exist at least $k$ sets $S_i$ with disjoint elements.      
\end{proof}

\section{Fixed Haircut Rates}
\label{app:fixedAlpha}

\subsection{Single Claims Trades}
When $\alpha \in [0,1]$ is fixed, the problem reduces to deciding whether the given trade is creditor-positive. This problem can be solved trivially in polynomial time whenever there is an oracle that can compute the clearing state $\vecp$ in polynomial time. First, check if $\alpha \ell_e \le a_w^x$, since otherwise the trade is impossible. Then call the oracle with networks $\calF$ and $\calF'$ (before and after the trade) and check if the assets fulfill Definition~\ref{def:creditPos}. Recall that the oracle for computing the clearing state can be implemented in polynomial time for proportional or edge-ranking payments. 

\subsection{Multi-Trades of Incoming Edges}

Consider a multi-trade of incoming edges. If the set $C$ of edges to be traded is given (along with a haircut rate $\alpha_i$ for each $e_i \in C$), then the problem is to decide whether the given trade is creditor-positive or not. This is again straightforward by computing the emerging clearing state and checking (1) whether $w$ can fund the return, $\sum_{e_i \in C} \alpha_i \ell_{e_i} \le a_w^x$, and (2) $v$ strictly profits after the trade, $a'_v > a_v$.

More interestingly, consider the problem of selecting the incoming edges. Suppose there is a vector of given haircut rates $\vecAlpha = (\alpha_e)_{e \in E^-(v)}$ for all incoming edges of $v$. We need to choose a set $C \subseteq E^-(v)$ of edges that yields a creditor-positive trade with maximal improvement for creditor $v$ (over all creditor-positive multi-trades of incoming edges of $v$ with given rate vector $\vecAlpha$). We first prove a hardness result for this problem. 

Similar to our result above, we show it is \classNP-hard to decide whether creditor $v$ can be \emph{saved} by a creditor-positive multi-trade of incoming edges, i.e., whether total assets of $L_v$ can be achieved. We call this problem \textsc{IncomingSave-FR} (for fixed haircut rates).

\begin{theorem}\label{thm:incoming-fixed-a-NP}
    \textsc{IncomingSave-FR} is weakly \classNP-hard.
\end{theorem}
\begin{proof}
    We use \subsetSum, where an instance consists of positive integers $S=\{b_1, \dots, b_n\}$ and a target value $T \leq \sum_{i=1}^n b_i$. W.l.o.g. we assume $b_i >1$, for all $i \in [n]$.

    For a given \subsetSum\ instance, we construct a corresponding financial network as follows. The graph $G$ consists of $n+3$ banks where $V = \{u, v, w, b_1, b_2, \dots, b_n\}$. Then, add edge $e_i$ with $\deb(e_i)=b_i$ and $\cre(e_i)=v$, for every $i \in [n]$. Set the liability $\ell_i$ to $b_i$. Note that we use $b_i$ to both denote bank $b_i$ as well as the respective integer value. Add edges $e$ and $e'$ with $\deb(e)=v, \cre(e)=w$ and $\deb(e')=u, \cre(e')=v$ with liabilities $\ell_{e}=T$ and $\ell_{e'}=T-1$. Banks $u$ and $w$ are the only banks with non-zero external assets $a^x_u = T-1$ and $a^x_w = T$. Finally, we fix $\alpha_{e_i}=1$, for $i \in [n]$, and $\alpha_{e'}=0$.

    We claim that there exists a subset $S' \subseteq [n]$ with sum of elements exactly $T$ if and only if there exists a creditor-positive multi-trade with $\vecAlpha$ that saves $v$.

    First, assume there exists a subset $S' \subseteq [n]$ with $\sum_{i \in S'} b_i = T$. For every $i \in S'$, add $e_i$ as well as $e'$ to the set $C$. Since $\sum_{e \in C} \alpha_e \cdot \ell_e = 0 + \sum_{e_i \in C\setminus \{e'\}} \alpha_{e_i} \cdot \ell_{e_i} = T$, the return of the multi-trade is given by $\rho=T = a^x_w$. Additionally due to the trade of edge $e'$, the return $\rho$ equals the total assets of $v$. For this return, $v$ is able to exactly settle all debt and pays all assets back to $w$. Hence, the multi-trade is creditor-positive and $v$ obtains total assets of $T = L_v$.

    For the other direction, assume there exists a trade that saves $v$. As no cycles can evolve by a multi-trade, the missing assets of $v$ must be received by the return. Since there exists no edge $e_i$ with $\alpha_{e_i} \cdot \ell_{e_i}=1$, edge $e'$ must be part of the trade. This reduces the incoming payments of $v$ to zero and implies that all assets of $T$ must stem from the return $\rho=T$. Thus, there exists a set of incoming edges $C \subseteq E^-(v)$ such that $\rho = \sum_{e \in C} \alpha_{e} \cdot \ell_e = 0 + \sum_{e_i \in C\setminus \{e'\}} \alpha_{e_i} \cdot \ell_{e_i} = \sum_{e_i \in C\setminus \{e'\}} \ell_{e_i} \geq T$. Since $w$ must not suffer from the trade, all funds transferred to $v$ must return to $w$, which implies $\rho = \sum_{e_i \in C\setminus\{e'\}} \ell_{e_i} \le T$. Hence, $S' = \{i \mid e_i \in C\setminus\{e'\}\}$ is a solution to \subsetSum.
\end{proof}

The proof can easily adapted to the special case with uniform fixed haircut rates $\alpha_e = \alpha \in [0,1]$ for all $e \in E^-(v)$, by removing bank $u$ and edge $e'$.

For the constructed network in proof of Theorem~\ref{thm:incoming-fixed-a-NP}, consider any creditor-positive multi-trade. Since $v$ is strictly profiting, her total assets must increase to $a'_v \geq T$. On the other hand, $w$ is not suffering and hence receives the full return $\rho$ as incoming payments. This implies $a'_v \leq T$. Combing both observations, every creditor-positive multi-trade saves $v$.

\begin{corollary}
    For a given financial network with banks $v$ and $w$ and fixed haircut rates, deciding whether a creditor-positive multi-trade of incoming edges exists is weakly \classNP-hard.
\end{corollary}

Towards approximation, we proceed roughly as for variable haircut rates above, with some notable differences. We lack the freedom to choose a suitable return for a set of edges. Instead, the return is fixed for each edge. Formally, we are given fixed haircut rates $\alpha_e$ for each $e \in E^-(v)$, so each edge has an associated return of $\rho_e = \alpha_e \ell_e$. Towards approximation, in Definition~\ref{def:approxMultiTrade} we relaxed the liability condition. Here, such a relaxation is not effective because of absent monotonicity properties. Instead, we resort to an approximate trade involving subsidies. Below, we formulate an FPTAS for approximate trades with limited subsidies. The variant with fixed returns again yields a certain \knapsack\ problem. However, upon rounding objective function parameters we alter return payments, which might destroy existence of any feasible trade.\footnote{As an example, consider a network where $v$ has $p$ incoming edge $e_i$ with debtor $u_i$, for $i \in [p]$, with sufficiently large liabilities. $v$ has an outgoing edge to $w$ with liability $q$. $u_1$ and $w$ are the only banks with non-zero external assets of $a^x_{u_1}=q-1$ and $a^x_w=q$. The fixed returns are given by $\rho_{e_1}=0, \rho_{e_2}=q$ and $\rho_{e_3}=\rho_{e_4}=\cdots=\rho_{e_p}=0$. Trading $C=\{e_1,e_2\}$ is a feasible solution. Clearly there exists an $\epsilon$ and $p$ such that $\lceil r_{e_2} / s\rceil \cdot s> q$ for $s=(\epsilon\cdot q)/p$.} As a consequence, we rely on two roundings tailored to $v$ and $w$, where rounding errors can be interpreted as small subsidies to support a trade.

More formally, suppose there is an \emph{exact trade} of edge set $C$ achieving asset level \emph{exactly} $A$ for $v$. By repeating our arguments for Lemma~\ref{lem:knapsack} we see that the return must be equal to the differences in assets emerging at both $v$ and $w$, i.e.,
\[
   \sum_{e \in C} \rho_e = d_v^{\spl} = d_w^{\spl} > 0 \; .
\]
Therefore, equation \eqref{eq:feasible} must again hold, which is independent of the traded set $C$ -- but depends on the asset level $A$ and the resulting payments $p'_e$. Unfortunately, fixed returns destroy the monotonicity property w.r.t.\ $A$ from Lemma~\ref{lem:monoton}, since the property crucially relies on adjusting returns for smaller values of $A$ (c.f.\ \ref{thm:incoming-fixed-a-NP}). This also complicates searching for an optimal asset level $A$, since our adjustment of binary search is not applicable.

Instead, for our FPTAS, we define a step size $s = \sum_{e \in E^{-}(v)} \ell_e \cdot \delta/n$, for any $\delta > 0$. We iterate over all asset values $\hat{A} = a_v + i\cdot s$, for $i = 1,\ldots, n / \delta$, which is polynomial in the network size and $1/\delta$. Note that we can set $\delta$ as a \emph{polynomially} small precision parameter (in contrast to exponentially small precision via binary search for variable returns above).

Now suppose there exists an exact trade with asset level $A$. Consider $\hat{A}$ as the unique value with $A \le \hat{A} < A + s$, and $\hat{p}'$ as the payments in $\calF^{sp}$ with outgoing assets $\hat{A}$ for $v$. Due to monotonicity, for all $e \in E$ we have $\hat{p}'_e \ge p'_e$, the payments resulting from outgoing assets $A$ for $v$. Thus,
\begin{align*}
   \sum_{e \in C} \rho_e = d_v^{\spl} = A - \left(a_v^x + \sum_{e \in E^{-}(v) \setminus C} p_e'\right) \; \ge \; \hat{A} - s - \left(a_v^x + \sum_{e \in E^{-}(v) \setminus C} \hat{p}'_e\right)
\end{align*}
and 
\begin{align*}
   \sum_{e \in C} \rho_e = d_w^{\spl} = \left(a_w^x + \sum_{e \in E^{-}(w) \cup C} p_e'\right) - a_w \; \le \; \left(a_w^x + \sum_{e \in E^{-}(v) \setminus C} \hat{p}'_e\right) - a_w
\end{align*}
Now let $r_e = \rho_e - \hat{p}'_e$ be the net payments\footnote{Observe that $C$ might be required to include edges for which $r_e$ is negative! \ref{thm:incoming-fixed-a-NP}} received by $v$ from $w$ for $e \in C$. Then
\begin{align}
    \sum_{e \in C} r_e &\ge \hat{A} - s - a_v^x - \sum_{e \in E^{-}(v)} \hat{p}_e' \label{eq:fixedConstr1} \\
    \sum_{e \in C} r_e &\le a_w^x + \sum_{e \in E^{-}(w)} \hat{p}_e' - a_w, \label{eq:fixedConstr2} \text{ and} \\
    \sum_{e \in C} \rho_e &\le a_w^x, \label{eq:fixedConstr3} \text{ the standard bound on return payments.}
\end{align}
When considering a rounded asset level $\hat{A}$, an exact trade for that asset level might not exist. Our construction shows that if an exact trade exists for level $A$, then an \emph{approximate trade} exists for the next-higher level $\hat{A}$. An approximate trade is one that satisfies \eqref{eq:fixedConstr1}-\eqref{eq:fixedConstr3}. It achieves asset level $\hat{A}$ for $v$, but only assets of at least $\hat{A}-s$ must be obtained using return payments from $w$ (with a total amount of at most $a_w^x$). Thus, a set $C$ that satisfies~\eqref{eq:fixedConstr1}-\eqref{eq:fixedConstr3} represents an exact trade after a subsidy of (at most) $s$.

For fixed asset level $\hat{A}$, the system \eqref{eq:fixedConstr1}-\eqref{eq:fixedConstr3} gives rise to a \knapsack\ problem\footnote{More precisely, the system \eqref{eq:fixedConstr1}-\eqref{eq:fixedConstr3} represents the decision problem of an instance of the 2-dimensional \knapsack\ problem.} and remains \classNP-hard to solve. Towards a Level-FPTAS, we introduce a second step size $s' = \sum_{e \in E^{-}(v)} |r_e| \cdot \delta/n^2$. Observe that $n \cdot s' \le s$. Then, for each $e \in E$, we round all values of $r_e$ (1) up to the next multiple of $s'$ to obtain $r_e^+$, and (2) down to the next multiple of $s'$ to obtain $r_e^-$. For simplicity, the Level-FPTAS again follows closely the standard \knapsack\ template. We fill a dynamic programming table with entries for all triples in
\begin{align*}
    \{1,2,\ldots,|E^{-}(v)|\} &\times \left\{-\sum_{e \in E^{-}(v)} |r_e^+/s'|,\ldots,-1,0,1,\ldots,\sum_{e \in E^{-}(v)} |r_e^+/s'|\right\}\\
    &\times \left\{-\sum_{e \in E^{-}(v)} |r_e^-/s'|\ldots,-1,0,1,\ldots,\sum_{e \in E^{-}(v)}|r_e^-/s'|\right\}
\end{align*}
This table has $O(n \cdot (n^2/\delta)^2)$ many entries. In each entry $DP(i,x,y)$, we store the minimum amount of return payment required to implement any trade with edges from $C \subseteq \{e_1,\ldots,e_i\} \subseteq E^{-}(v)$ that has $\sum_{e \in C} r_e^+ = x\cdot s'$ and $\sum_{e \in C} r_e^- = y\cdot s'$. The update of entry $(i,x,y)$ is clearly
\[
DP(i,x,y) = \min\{ DP(i-1,x-(r_{e_i}^+/s'),y-(r_{e_i}^-/s')) + \rho_{e_i}, DP(i-1,x,y)\} 
\]
with a base case of $DP(1,(r_{e_1}^+/s'),(r_{e_1}^-/s')) = \rho_{e_1}$ and $DP(1,x,y) = \infty$ otherwise. We search for a solution by inspecting the entries $(|E^{-}(v)|,x,y)$ such that $x\cdot s' \ge  \hat{A} - s - a_v^x - \sum_{e \in E^{-}(v)} \hat{p}_e'$ and $y\cdot s' \le a_w^x + \sum_{e \in E^{-}(w)} \hat{p}_e' - a_w$ for one that has value at most $a_w^x$. If there is such an entry, then we found an approximate solution to \eqref{eq:fixedConstr1} and \eqref{eq:fixedConstr2} that exactly satisfies \eqref{eq:fixedConstr3}. The corresponding trade $C$ violates the each of the constraints \eqref{eq:fixedConstr1} and \eqref{eq:fixedConstr2} by at most $|C|\cdot s' < s$. As such, we need to subsidize $v$ by at most another $s$ and subsidize $w$ by at most $s$ to obtain feasible conditions for an exact trade w.r.t.\ $\hat{A}$. 

Hence, a solution $C$ for the Level-FPTAS for a given asset level $\hat{A}$ yields an exact trade with asset level at least $\hat{A}$ for $v$ when we give subsidies of at most $2s$ to $v$ and at most $s$ to $w$. We call such a trade \emph{$\delta$-subsidized}. Now if the instance allows an exact trade $C$ with optimal assets $A^*$ for $v$, then $C$ also represents a $\delta$-subsidized trade with the corresponding subsidies and asset level $\hat{A} \ge A^*$ for $v$. As such, by running our Level-FPTAS for all (polynomially many) values of $\hat{A}$, we obtain a 0-optimal $\delta$-subsidized trade. The running time is clearly polynomial in the representation of $\calF$ and $1/\delta$.

\begin{theorem}
    Consider a financial network with monotone payment functions and efficient clearing oracle, creditor $v$ and buyer $w$. If there exists a creditor-positive multi-trade of incoming edges, then an $\delta$-subsidized 0-optimal trade with subsidies can be computed in time polynomial in the size of $\calF$, and $1/\delta$, for every $\delta > 0$.
\end{theorem}

\subsection{Multi-Trades of Outgoing Edges}

Consider a financial network with banks $u$ and $w$ where every outgoing edge $e\in E^+(u)$ of $u$ is assigned a predetermined haircut rate $\alpha_e \in [0,1]$. For a given subset of outgoing edges $C \subseteq E^+(v)$, the goal is to decide whether $C$ together with $\vecAlpha=(\alpha_i)_{e_i\in C}$ is a Pareto-positive multi-trade. If the clearing state of the network can be derived efficiently, e.g., via a given efficient clearing oracle, the problem lies in \classP. First, verify if $w$ is able to fund the return, i.e., $a^x_w \geq \sum_{e_i \in C} \alpha_i \cdot\ell_{e_i}$. Then, perform the given multi-trade, derive the resulting clearing state $\vecp$ and check whether $w$ is not harmed while all creditors strictly profit, i.e., $a'_w \geq a_w$ and $a'_{v_i} > a_{v_i}$ for all $v_i$ with an incoming edge $e_i \in C$.

The problem becomes more challenging if the set of trading edges is not fixed but can be chosen as part of the solution. Formally, we are given a financial network with monotone payment functions and efficient clearing oracle. Consider banks $u$ and $w$ and haircut rates $\vecAlpha$ for all outgoing edges of $u$. The objective is to compute a set $C \subseteq E^+(v)$ such that (1) the multi-trade $C$ together with $\vecAlpha$ is Pareto-positive and (2) maximizes the sum of profits of $u$'s creditors. This problem is denoted by \textsc{Outgoing-FR}.

\begin{theorem}
    \textsc{Outgoing-FR} is strongly \classNP-hard. For any constant $\varepsilon > 0$, there exists no efficient $n^{1/2-\varepsilon}$-approximation algorithm unless \classP\ = \classNP.
\end{theorem}
\begin{proof}
    We show the statement by slightly adapting the proof of Theorem~\ref{thm:setpacking}. 

    For an instance of \setPacking, construct a corresponding financial network as described in proof of Theorem~\ref{thm:setpacking} and depicted in Figure~\ref{fig:setpacking}. Fix all haircut rates to 1. Then there exists a feasible set packing of size at least $k$ if and only if there exists a Pareto-positive multi-trade in the corresponding network where $v$'s creditors profit by at least $kM$ in total.  

    Conceptually, a fixed haircut rate of 1 implies that each creditor $S_i$ of a traded edge fully pays all outgoing edges, independent of the payment function used for clearing. As such, sufficiently large payments can only feasibly flow back throught the gadget to $w$ if the traded edges correspond to non-overlapping sets. This is the key property that was implemented in Theorem~\ref{thm:setpacking} by placing an edge of large liability $M$ towards $w$ in the last position of the edge ranking of each bank $S_i$. 
    
    As a consequence, the arguments in the proofs of Theorem~\ref{thm:setpacking} can be repeated almost verbatim to obtain the same results for fixed haircut rates and arbitrary monotone clearing.
\end{proof}

\end{document}